\documentclass[12pt]{article}
\usepackage[utf8]{inputenc}
\usepackage{lmodern}
\usepackage[T1]{fontenc}

\usepackage{graphicx}
\usepackage{fullpage}
\usepackage{amsmath,amsthm,amssymb,mathtools}

\usepackage{float}
\usepackage{bbm}
\usepackage{enumitem}

\usepackage{hyperref}
\hypersetup{
colorlinks=true,
linkcolor=blue,
filecolor=blue,
urlcolor=blue,
citecolor=blue,
}

\usepackage{tikz}
\usepackage{pgfplots}
\usetikzlibrary{arrows.meta}
\usetikzlibrary{arrows}
\usetikzlibrary{decorations.pathreplacing}

\pgfplotsset{compat=1.18} 

\usepackage{natbib}
\setcitestyle{authoryear,round} 

\usepackage{caption}

\usepackage{setspace}

\usepackage[authormarkup=none]{changes}
\definechangesauthor[name=Alp, color=red]{alp}

\usepackage[margin=3.2cm]{geometry}
\usepackage{titlesec}
\titleformat*{\section}{\large\bfseries}
\titleformat*{\subsection}{\bfseries}

\def\RR{\mathbb{R}}	

\newtheorem{thm}{Theorem}
\newtheorem{lemma}{Lemma}

\newtheorem{example}{Example}

\newtheorem{prop}[thm]{Proposition}

\newtheorem{innercustomthm}{Theorem}
\newenvironment{customthm}[1]
  {\renewcommand\theinnercustomthm{#1}\innercustomthm}
  {\endinnercustomthm}

\newcommand{\overbar}[1]{\mkern 1.5mu\overline{\mkern-1.5mu#1\mkern-1.5mu}\mkern 1.5mu}

\begin{document}
\onehalfspacing 

\thispagestyle{empty}
\def\thefootnote{\fnsymbol{footnote}}

\begin{center}
{\Large\bf Vocabulary aggregation\footnote{We thank Ehsan Ali, Maria Rifqi, Marco Tolotti and audiences in Catania, Paris, Turin and Venice for their comments. An inspiring talk delivered by Alan Miller at the University of Florence in September 2023 sparked our curiosity. The first author received funding from the European Union Next-Generation EU National Recovery and Resilience Plan (NRRP) under the PRIN project 20222Z3CR7, ``Nudging under Limited Attention''. The second author received funding from the European Union's Horizon 2020 research and innovation programme under the Marie Skłodowska-Curie grant agreement no.\ 956107, ``Economic Policy in Complex Environments''.}}
\end{center}

\vspace*{0.5cm}
\begin{center}
\begin{tabular}{ccc}
{\sc Marco LiCalzi} & $\quad$ & {\sc M.\ Alperen Yasar}\\
{\small Universit\`{a} Ca' Foscari Venezia} & &  {\small Institut des Sciences Cognitives Marc Jeannerod, CNRS} \\
{\small\tt licalzi@unive.it} & & {\small\tt alperen.yasar@isc.cnrs.fr}\\
\end{tabular}
\end{center}
\vspace*{0.5cm}
\begin{center}
February 2026
\end{center}
\vspace*{0.4cm}

{\small
\noindent {\bf Abstract.} A vocabulary is a list of words designating subsets from a grand set $X$. We model a vocabulary as a partition of $X$ and study the aggregation of individual vocabularies into a collective one. We characterize aggregation rules when $X$ is linearly ordered and each word of the vocabulary spans an order interval. We allow for individual vocabularies to differ both in the number and in the span of their words. Under a suitable restriction on agents' preferences, we show that our aggregation rules are strategy-proof.
}



\setcounter{footnote}{0}
\def\thefootnote{\arabic{footnote}}

\newpage

\section{Introduction}

An ordinal structure is often communicated using a coarse \emph{grading scale}. Restaurants, movies, and many other goods are commonly rated. Students' (and teachers') performance is given letter or numerical grades. \citet[Chapter~7]{bal10} collect several examples and historical tidbits; for instance, the mathematician Monge used grades from $a$ to $g$ during his tenure (1783--1790) as examiner for the French navy, and Harvard introduced the five-letter grading scale in 1883. 

Legal codes have statutes rating crimes by their severity. Section~2501 of the 2023 Pennsylvania Statute codifies five different categories of homicide; see \cite{kee49} on the motivation for this refinement. Similarly, Section~775.081 of the 2023 Florida Statute  classifies felonies into five categories and misdemeanors into two categories.

Ordinal classifications abound also in biology and medicine. For instance, the Apgar score summarizes the health of newborn infants using 11 ratings; see \cite{mis15} for more examples. Psychology, especially psychometrics, is concerned with the measurement of constructs often defined on nominal or ordinal scales \citep{ste46}. This paper focuses on ordinal scales.

An important distinction is whether the grading scale is the coarsening of a  (perhaps, latent) linear order that graders' elicitations adhere to. In many applications, there is no such intrinsic order: different agents may come up with ordinal evaluations based on different premises and rank objects inconsistently --- that is, one puts $a$ above $b$ while another ranks $b$ over $a$. A lack of consensus on the underlying linear order opens up significant hurdles when grading is carried out by groups or committees \citep{mor16}.

Other applications, however, are based on a (reasonably) objective ordering. The 1-minute Apgar score assesses how well a baby has tolerated the birthing process. Two doctors may assign babies with different Apgar scores, but their relative grades are expected to be consistent. The statutes strongly suggest judges or juries rank crimes only by their severity, to match them with proportionate punishments. When a latent order exists, it may be learned from preference judgements \citep{coe97}. 

We focus on the problem of aggregating individual classifications based on a common linear order into a consistent collective classification. See \cite{fal20} for a related approach to collective evaluations. Such aggregation rules are useful in statistics and machine learning, where ensemble methods \citep{zho12} combine multiple learning algorithms to obtain superior performance in classification or prediction tasks involving ordinal structures \citep{bel20}. 

Imagine an agent tasked with partitioning an interval of the real line into subintervals. Each subinterval is labeled and interpreted as a \emph{word}. Because subintervals inherit the intrinsic order of the real line, words are naturally ordered from smaller to higher ranks. Then the individual partition defines a grading scale. For instance, an agent may rate `small' any value below $t_1$ and `intermediate' any value between $t_1$ and $t_2$. Another agent may rate `small' any value below $t^\prime$, with $t_1 < t^\prime < t_2$. The two agents use the word 'small' differently, but they agree that $t_1 < t^\prime < t_2$.

An individual grading scale is called a \emph{vocabulary}. Thus, we are interested in the aggregation of individual vocabularies into a collective vocabulary.\footnote{\cite{loe12} offer a broader perspective where common vocabularies are instrumental in the social construction of meaning.} Following the abstract approach of \cite{wil75} and \cite{rub86}, we study aggregation rules that map a profile of individual vocabularies into a collective vocabulary, using the assumption of a common linear order to leverage results from the theory of aggregation on ordinal scales \citep[Chapter~8]{gra09}. 

Our construction of the collective vocabulary is a special case of the general problem of choosing a common partition given a profile of individual partitions \citep{bar81} or equivalence relations \citep{fis86}. This issue has been recently revisited under the assumption that agents have preferences over the common partition, studying its strategy-proofness in general \citep{mis12} and in the context of budget proposals \citep{fre21}. Finally, \citet{kla20} and \citet{bha24} consider agents with single-peaked preferences who engage in social choice about the aggregation of contiguous objects into intervals. 

Our work is inspired by \cite{mil25}, who studies the aggregation of intervals (equivalent to minimal vocabularies) in reference to the construction of community standards, where the goal is to ascertain the bounds within which an action is deemed reasonable, or a statement is not defamatory. 

Section~\ref{sec2} sets up the model. Section~\ref{sec3} starts with the simple case of \emph{homogeneous} vocabularies where all the agents use the same set of words, but may attach them to different subintervals. We generalize the median rule used by \cite{blo10} and \citet{far11} and the endpoint rules proposed in \cite{mil25} from intervals to vocabularies. Section~\ref{sec4} scales up to the case of \emph{heterogeneous} vocabularies where agents may neglect some words; for example, an agent uses only the opposition `white-black' while another has a nuanced vocabulary with a few intermediate shades of grey.

Section~\ref{sec5} deals with the case of \emph{incomplete} homogeneous vocabularies, when agents do not submit their classifications but merely rate a few exemplars in accordance with their individual vocabularies. We show how to derive a collective vocabulary that is usually incomplete but incorporates all the information provided by agents' classified exemplars. Section~\ref{sec6} introduces the assumption that agents have preferences over the collective vocabulary, and may try to manipulate it by misreporting their individual vocabularies.  We consider a class of individual preferences under which strategy-proofness is a key property that characterizes the same aggregation rules derived in the earlier sections.

\section{Setup}\label{sec2}

\subsection{An example}\label{sec21}

We introduce the model with an illustrative example, which we later return to.

\begin{example}\label{ex1}
There is a committee of three faculty members who must agree on a grading scheme for students. Students' exams are scored on a scale ranging from $0$ to $100$, but each student is to receive one of the five grades $F|D|C|B|A$. Each professor issues an individual recommendation for mapping scores to grades. 

Primus (Agent~1) recommends equally spaced intervals, which we represent in a visual diagram:
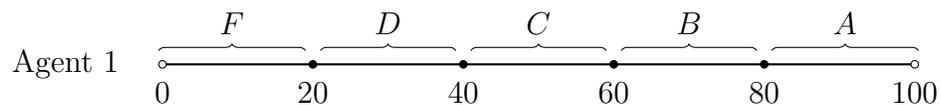
\begin{figure}[H]\centering
	\begin{tikzpicture}[scale=0.1]
		\node[text width=2cm] at (-10,30) {Agent $1$};
		\draw[decorate, decoration={brace}, yshift=10ex]  (-0.5,30) -- node[above=0.4ex] {$F$}  (19,30);
		\draw[decorate, decoration={brace}, yshift=10ex]  (21,30) -- node[above=0.4ex] {$D$}  (39,30);
		\draw[decorate, decoration={brace}, yshift=10ex]  (41,30) -- node[above=0.4ex] {$C$}  (59,30);
		\draw[decorate, decoration={brace}, yshift=10ex]  (61,30) -- node[above=0.4ex] {$B$}  (79,30);
		\draw[decorate, decoration={brace}, yshift=10ex]  (81,30) -- node[above=0.4ex] {$A$}  (100.5,30);
		\draw[-, thick] (0.5,30) -- (20,30);
		\draw[fill=white] (0,30) circle (5mm) node[below=1mm] {$0$};
		\draw[fill=black] (20,30) circle (5mm) node[below=1mm] {$20$};
		\draw[fill=black] (40,30) circle (5mm) node[below=1mm] {$40$};
		\draw[fill=black] (60,30) circle (5mm) node[below=1mm] {$60$};
		\draw[fill=black] (80,30) circle (5mm) node[below=1mm] {$80$};
		\draw[-, thick] (20,30) -- (40,30);
		\fill (20,30) circle (5mm);
		\draw[-, thick] (40,30) -- (60,30);
		\fill (40,30) circle (5mm);
		\draw[-, thick] (60,30) -- (80,30);
		\fill (80,30) circle (0.25);
		\draw[-, thick] (80,30) -- (100.5,30);
		\draw[fill=white] (100,30) circle (5mm) node[below=1mm] {$100$};
	\end{tikzpicture}
\end{figure}

\noindent
The diagram is unclear about which grade an interior (black) endpoint maps to. However, by convention, ambiguities are solved in favor of students: each black endpoint is attributed to the grade on its right. Hence, Primus' recommendation maps $[0,20)$ to $F$, $[20,40)$ to $D$, and so on.

Secunda (Agent~2) is more lenient than Primus, and recommends a generous scale:
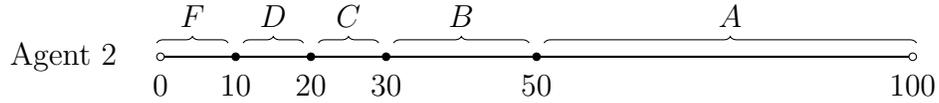
\begin{figure}[H]\centering
	\begin{tikzpicture}[scale=0.1]
		\node[text width=2cm] at (-10,15) {Agent $2$};
		\draw[decorate, decoration={brace}, yshift=10ex]  (-0.5,15) -- node[above=0.4ex] {$F$}  (9,15);
		\draw[decorate, decoration={brace}, yshift=10ex]  (11,15) -- node[above=0.4ex] {$D$}  (19,15);
		\draw[decorate, decoration={brace}, yshift=10ex]  (21,15) -- node[above=0.4ex] {$C$}  (29,15);
		\draw[decorate, decoration={brace}, yshift=10ex]  (31,15) -- node[above=0.4ex] {$B$}  (49,15);
		\draw[decorate, decoration={brace}, yshift=10ex]  (51,15) -- node[above=0.4ex] {$A$}  (100.5,15);
		\draw[-, thick] (0.5,15) -- (10,15);
		\draw[fill=white] (0,15) circle (5mm) node[below=1mm] {$0$};
		\draw[fill=black] (10,15) circle (5mm) node[below=1mm] {$10$};
		\draw[-, thick] (10,15) -- (20,15);
		\draw[fill=black] (20,15) circle (5mm) node[below=1mm] {$20$};
		\draw[-, thick] (20,15) -- (30,15);
		\draw[fill=black] (30,15) circle (5mm) node[below=1mm] {$30$};
		\draw[-, thick] (30,15) -- (50,15);
		\draw[fill=black] (50,15) circle (5mm) node[below=1mm] {$50$};
		\draw[-, thick] (50,15) -- (100.5,15);
		\draw[fill=white] (100,15) circle (5mm) node[below=1mm] {$100$};
	\end{tikzpicture}
\end{figure}

Tertium (Agent~3) prefers to cluster the extremes and recommends:
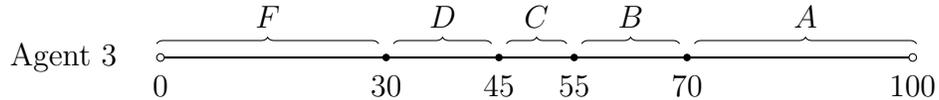
\begin{figure}[H]\centering
	\begin{tikzpicture}[scale=0.1]
		\node[text width=2cm] at (-10,00) {Agent $3$};
		\draw[decorate, decoration={brace}, yshift=10ex]  (-0.5,0) -- node[above=0.4ex] {$F$}  (29,0);
		\draw[decorate, decoration={brace}, yshift=10ex]  (31,0) -- node[above=0.4ex] {$D$}  (44,0);
		\draw[decorate, decoration={brace}, yshift=10ex]  (46,0) -- node[above=0.4ex] {$C$}  (54,0);
		\draw[decorate, decoration={brace}, yshift=10ex]  (56,0) -- node[above=0.4ex] {$B$}  (69,0);
		\draw[decorate, decoration={brace}, yshift=10ex]  (71,0) -- node[above=0.4ex] {$A$}  (100.5,0);
		\draw[-, thick] (-0.5,0) -- (30,0);
		\draw[fill=white] (0,0) circle (5mm) node[below=1mm] {$0$};
		\draw[-, thick] (30,0) -- (45,0);
		\fill[black] (30,0) circle (5mm) node[below=1mm] {$30$};
		\fill[black] (45,0) circle (5mm) node[below=1mm] {$45$};
		\draw[-, thick] (45,0) -- (55,0);
		\fill[black] (55,0) circle (5mm) node[below=1mm] {$55$};
		\draw[-, thick] (55,0) -- (70,0);
		\fill[black] (70,0) circle (5mm) node[below=1mm] {$70$};
		\draw[-, thick] (70,0) -- (100.5,0);
		\draw[fill=white] (100,0) circle (5mm) node[below=1mm] {$100$};
	\end{tikzpicture}
\end{figure}
\end{example}

The question discussed in this paper is the aggregation of individual recommendations into a common grading scheme.

\subsection{The model}\label{sec22}

We consider an open interval $X$ from the linearly ordered set $(\RR,<)$. Assuming that $X$ is open simplifies some technicalities, but our results extend to any nonempty interval. 

An \emph{interval partition} $\Pi$ for $X$ is formed by a finite number of intervals delimited by endpoints. Intuitively, the intervals enact the partition but their endpoints are an effective way to identify it. For instance, Primus' grading scheme in Example~\ref{ex1} is the interval partition for $X=[0,100]$ associated with the endpoints $20, 40, 60, 80$. Conveniently, we identify an interval partition with the list of all its (internal) endpoints. 

An interval partition $\Pi$ categorizes $X$ into distinct (but yet unnamed) classes. If we add a distinct label for each interval, we have a \emph{labeled partition}; e.g., Primus' grading scheme attaches the label $F$ to the interval $[0,20)$, the label $D$ to the interval $[20, 40)$, and so on.

A labeled partition is a \emph{vocabulary}, where each item is a label (called \emph{word}) supported by an interval (called \emph{extent}\footnote{We follow the convention that the extent supporting a word includes the left endpoint, but other approaches are possible.}). Figure~\ref{fig-ex2} depicts a 3-item vocabulary for the interval $(0,1)$, where the three words $w_1, w_2, w_3$ have extents $(0,a), [a,b), [b,1)$, respectively.
\begin{figure}[!ht]
\vspace{2mm}
\centering
\begin{tikzpicture}
	\draw[>={Latex[width=4mm]}, thick] (0, 0) -- (9, 0);
	\draw[fill=white] (0,0) circle (0.75mm) node[below=2mm] {$0$};
	\draw[decorate, decoration={brace}, yshift=1ex]  (0.1,0) -- node[above=0.4ex] {$w_1$}  (2.9,0);
	\fill[black] (3,0) circle (0.75mm) node[below=2mm] {$a$};
	\draw[decorate, decoration={brace}, yshift=1ex]  (3.1,0) -- node[above=0.4ex] {$w_2$}  (6.9,0);
	\fill[black] (7,0) circle (0.75mm) node[below=2mm] {$b$};
	\draw[decorate, decoration={brace}, yshift=1ex]  (7.1,0) -- node[above=0.4ex] {$w_3$}  (8.9,0);
	\draw[fill=white] (9,0) circle (0.75mm) node[below=2mm] {$1$};
\end{tikzpicture}
\caption{A vocabulary}\label{fig-ex2}
\end{figure}
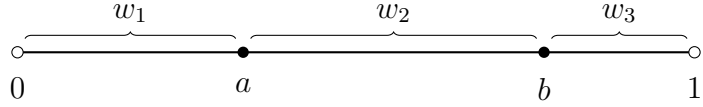
Following the natural order on extents induced by the linear order on $X$, we may interpret the three words as `small', `intermediate', `large'. 

Let $N = \{1,2, \ldots, n\}$ be a finite set of agents. Each agent $i$ has a vocabulary $V_i$ based on the same set of words. An \emph{aggregation rule} for vocabularies $f: V^n \rightarrow V$ maps each profile of $n$ individual vocabularies into a collective vocabulary.

\section{Homogeneous vocabularies}\label{sec3}

A word is \emph{active} for a vocabulary if it has a nonempty extent. We begin with the simpler case of homogeneous vocabularies where each agent has the same active words, although different agents may have different extents for the same word.  

Given $n$ homogeneous individual vocabularies with $m+1$ words, an aggregation rule decides the extents of the $m+1$ words for the collective vocabulary. We focus on aggregation rules that generalize the \emph{endpoint rules} introduced by \citet{mil25}.

Let $s_i = (s_i^1, \ldots, s_i^k, \ldots, s_i^m)$ be the sequence of endpoints for the individual vocabulary $V_i$, arranged in increasing order. Words are also arranged in increasing order: the extent for the first word is $(0,s_i^1)$ and for the $k$-th word is $[s_i^{k-1},s_i^k)$.

An \emph{aggregation rule} $f$ maps the profile $\textbf{s} = \{ s_i \}_{i=1}^n$ of individual endpoint sequences to an ordered sequence $s_c = (s_c^1, \ldots, s_c^k, \ldots, s_c^m)$ of collective endpoints that identifies the collective vocabulary. It is understood that an aggregation rule $f$ must be \emph{consistent}, generating $m+1$ ordered words. Thus, we require that the collective endpoints form a nondecreasing sequence: given $f(\textbf{s})=\left( f^1(\textbf{s}), \ldots, f^m(\textbf{s}) \right)$, we must have $f^1(\textbf{s}) \le \ldots \le f^k(\textbf{s}) \le \ldots \le f^m(\textbf{s})$ for any profile $\textbf{s}$.

We say that the aggregation rule is (component-wise) \emph{separable} if the collective $k$-th endpoint $s_c^k$ depends only on the $n$ individual $k$-th endpoints $(s_1^k, \ldots , s_n^k)$; that is, the aggregation rule $f(\textbf{s})=\left( f^1(\textbf{s}), \ldots, f^m(\textbf{s}) \right)$ may be decomposed into $m$ distinct functions such that  $f^k(\textbf{s}) = f^k (s_1^k, \ldots , s_n^k)$. 

The following are three consistent and separable endpoint rules for homogeneous vocabularies. 

\begin{description}
\item[Average rule:] for each $k$, $f^k_{ave}(\textbf{s}) = \left(\sum_{i=1}^n s_i^k\right)/n$.

\item[Median rule:] for each $k$, $f^k_{med}(\textbf{s}) = \texttt{med} \left(s_1^k, \ldots , s_n^k\right)$, where the median operator \texttt{med} selects the central value from any ordered sequence $(s_1, \ldots, s_n)$, assuming that $n$ is odd.

\item[Dictatorship of agent $i$:] for each $k$, $f^k_{(i)}(\textbf{s})=s^k_i$.

\end{description}

Applied to Example~\ref{ex1}, the mean rule selects the endpoints $\{20, 35, 48.33, 66.67\}$. Note that the last three endpoints are not present in any individual vocabulary. By construction, instead, the other two rules pick only endpoints that appear in at least one individual vocabulary. Specifically, the median rule yields the endpoints $\{20, 40, 55, 70\}$, while the dictatorship of (say) Agent~2 selects the endpoints $\{10, 20, 30,  50\}$. Figure~\ref{fig:aggreg} shows the different vocabularies selected by the three rules.

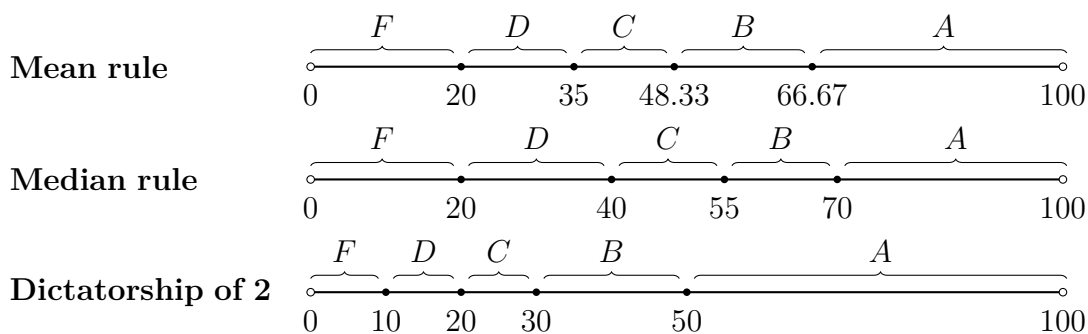
\begin{figure}[!htbp]\centering
	\begin{tikzpicture}[scale=0.1]
		\node[text width=4cm] at (-20,-15) {\textbf{Median rule}};
		\node[text width=4cm] at (-20,0) {\textbf{Mean rule}};
		\node[text width=4cm] at (-20,-30) {\textbf{Dictatorship of 2}};
		
		\draw[decorate, decoration={brace}, yshift=10ex]  (0,0) -- node[above=0.4ex] {$F$}  (19,0);
		\draw[decorate, decoration={brace}, yshift=10ex]  (21,0) -- node[above=0.4ex] {$D$}  (34,0);
		\draw[decorate, decoration={brace}, yshift=10ex]  (36,0) -- node[above=0.4ex] {$C$}  (47.33,0);
		\draw[decorate, decoration={brace}, yshift=10ex]  (49.33,0) -- node[above=0.4ex] {$B$}  (65.67,0);
		\draw[decorate, decoration={brace}, yshift=10ex]  (67.77,0) -- node[above=0.4ex] {$A$}  (100.5,0);
		\draw[-, thick] (-0.5,0) -- (20,0);
		\draw[fill=white] (0,0) circle (5mm) node[below=1mm] {$0$};
		\fill (20,0) circle (5mm) node[below=1mm] {$20$};
		\draw[-, thick] (20,0) -- (35,0);
		\fill (35,0) circle (5mm) node[below=1mm] {$35$};
		\draw[-, thick] (35,0) -- (48.33,0);
		\fill (48.33,0) circle (5mm) node[below=1mm] {$48.33$};
		\draw[-, thick] (48.33,0) -- (66.67,0);F
		\fill (66.67,0) circle (5mm) node[below=1mm] {$66.67$};
		\draw[-, thick] (66.67,0) -- (100.5,0);
		\draw[fill=white] (100,0) circle (5mm) node[below=1mm] {$100$};

		\draw[decorate, decoration={brace}, yshift=10ex]  (0,-15) -- node[above=0.4ex] {$F$}  (19,-15);
		\draw[decorate, decoration={brace}, yshift=10ex]  (21,-15) -- node[above=0.4ex] {$D$}  (39,-15);
		\draw[decorate, decoration={brace}, yshift=10ex]  (41,-15) -- node[above=0.4ex] {$C$}  (54,-15);
		\draw[decorate, decoration={brace}, yshift=10ex]  (56,-15) -- node[above=0.4ex] {$B$}  (69,-15);
		\draw[decorate, decoration={brace}, yshift=10ex]  (71,-15) -- node[above=0.4ex] {$A$}  (100.5,-15);

		\draw[-, thick] (-0.5,-15) -- (20,-15);
		\draw[fill=white] (0,-15) circle (5mm) node[below=1mm] {$0$};
		\fill (20,-15) circle (5mm) node[below=1mm] {$20$};
		\fill (40,-15) circle (5mm) node[below=1mm] {$40$};

		\draw[-, thick] (70,-15) -- (100.5,-15);
		\fill (55,-15) circle (5mm) node[below=1mm] {$55$};
		\fill (70,-15) circle (5mm) node[below=1mm] {$70$};
		\draw[fill=white] (100,-15) circle (5mm) node[below=1mm] {$100$};
		\draw[-, thick] (20,-15) -- (40,-15);
		\draw[-, thick] (40,-15) -- (55,-15);
		\draw[-, thick] (55,-15) -- (70,-15);

	\draw[decorate, decoration={brace}, yshift=10ex]  (0,-30) -- node[above=0.4ex] {$F$}  (9,-30);
\draw[decorate, decoration={brace}, yshift=10ex]  (11,-30) -- node[above=0.4ex] {$D$}  (19,-30);
\draw[decorate, decoration={brace}, yshift=10ex]  (21,-30) -- node[above=0.4ex] {$C$}  (29,-30);
\draw[decorate, decoration={brace}, yshift=10ex]  (31,-30) -- node[above=0.4ex] {$B$}  (49,-30);
\draw[decorate, decoration={brace}, yshift=10ex]  (51,-30) -- node[above=0.4ex] {$A$}  (100.5,-30);

\draw[-, thick] (-0.5,-30) -- (10,-30);
\draw[-, thick] (10,-30) -- (20,-30);
\draw[-, thick] (20,-30) -- (30,-30);
\draw[-, thick] (30,-30) -- (50,-30);
\draw[-, thick] (50,-30) -- (100,-30);

\draw[fill=white] (0,-30) circle (5mm) node[below=1mm] {$0$};
\fill (10,-30) circle (5mm) node[below=1mm] {$10$};
\fill (20,-30) circle (5mm) node[below=1mm] {$20$};

\fill (30,-30) circle (5mm) node[below=1mm] {$30$};
\fill (50,-30) circle (5mm) node[below=1mm] {$50$};
\draw[fill=white] (100,-30) circle (5mm) node[below=1mm] {$100$};

\end{tikzpicture}
	\caption{Collective vocabularies under three aggregation rules}\label{fig:aggreg}
\end{figure}

The three rules deliver quite different vocabularies, but there is some common ground. Notably, the three agents concur in recommending that $[0,10)$ deserves an $F$ grade and $[80,100]$ deserves an $A$ grade. This unanimity in individual vocabularies is upheld by each of the three rules.

A separable endpoint rule need not be consistent. For instance, in the above example, consider the separable aggregation rule where Agent~1 is the dictator for the first two endpoints and Agent~2 is the dictator for the last two endpoints. Then $f^1_{(1)} = 20$ and $f^2_{(1)} = 40$, but $f^3_{(2)} = 30$ and $f^4_{(2)} = 50$, with $20 < 40 \not< 30 < 50$.

Conversely, a consistent rule need not be separable. Here is an example.
\begin{description}
\item[Multiset rule:] first, arrange the $m$ individual endpoints of the $n$ agents into an ordered multiset (including repeated endpoints) with $m \cdot n$ elements; then, divide the ordered multiset into $m$ consecutive groups of $n$ elements and select the median of each group, assuming that $n$ is odd.
\end{description}

Returning to Example~\ref{ex1}, the multiset rule arranges the $4 \times 3 = 12$ individual endpoints in the (ordered) multiset $\{10, 20, \allowbreak 20 \> | \> \allowbreak 30, 30, 40,  \> | \> 45, 50, 55,  \> | \> 60, 70, 80\}$, where the bar separates the $m=4$ consecutive subgroups. Selecting the median from each subgroup, we obtain $\{20, 30, 50, 70\}$.

\subsection{Axioms}

Under homogeneity, each individual vocabulary is identified by an ordered sequence $s^1 < \ldots < s^m$ of endpoints. Given $n$ agents, we juxtapose the individual ordered sequences into a matrix-like arrangement 
$$M = \begin{bmatrix}
s_1^1 & s_1^2 & \cdots & s_1^m\\
s_2^1 & s_2^2 & \cdots & s_2^m\\
\vdots & \vdots & \vdots & \vdots \\
s_n^1 & s_n^2 & \cdots & s_n^m\\
\end{bmatrix}$$
where row $i$ contains the ordered sequence of endpoints for agent~$i$, and column~$k$ lists the $k$-th endpoints for all agents. We write $M_i$ and $M^k$ to denote the $i$-th row and the $k$-th column of $M$ and let ${\cal M}$ be the set of all profiles of individual endpoints.

An aggregation rule $f: {\cal M} \to X^m$ maps each profile of individual endpoints in $\cal M$ into a vector $s_c = (s_c^1, \ldots, s_c^k, \ldots, s_c^m)$ of collective endpoints in $X^m$. We focus on aggregation rules that are separable and consistent; that is, we assume
$$f(M)=\left( f^1(M^1), \ldots, f^m(M^m) \right)$$ 
with $f^1(M^1) \le \ldots \le f^k(M^k) \le \ldots \le f^m(M^m)$. We say that the function $f^k(M^k)$ is the $k$-th \emph{component} of the aggregation rule. 

We consider the following properties. The first one requires that the aggregation rule generates $m$ distinct collective endpoints. This implies that the collective vocabulary attributes a nonempty extent to each of the $m+1$ words; in short, all the original words remain active.

\smallskip\noindent
\textbf{Strict consistency.} For any $M$, $f^1(M^1) < \ldots < f^k(M^k) < \ldots < f^m(M^m)$.
\smallskip

The second property requires that, if all agents agree on the same $k$-th endpoint, this must be chosen as the $k$-th collective endpoint. Consequently, if each individual vocabulary attributes the same extent to a specific word, this carries over to the collective vocabulary.

\smallskip\noindent
\textbf{Unanimity.} If $s^k_i = s$ for all $i$, then $f^k(M^k) = s$.
\smallskip

The third property requires that agents are indistinguishable. Therefore, swapping rows in $M$ does not affect the outcome of the aggregation. Given a permutation $\pi$ over the set of agents $N$ and a profile $M$ in $\cal M$, denote by $\pi M$ in $\cal M$ the profile obtained by applying $\pi$ over the rows of $M$. We call it anonymity, but this property is often known as symmetry in the theory of aggregation functions.

\smallskip\noindent
\textbf{Anonymity.} For every permutation $\pi$ and any $M$, $f(M) = f(\pi M)$.
\smallskip

The fourth property requires that if all individual endpoints are modified by the same increasing transformation, then the collective endpoints change accordingly. Equivalently, all endpoints (individual or collective) are expressed in the same ordinal scale. This property is called weak neutrality by \citet{mil25} and ordinal covariance by \citet{cha07}. The literature on aggregation functions speaks of ordinal stability or ordinal scale invariance. Leaving ordinal implicit, we call it stability.

Let $\Phi^\uparrow$ be the set of all increasing bijections $\varphi: X \to X$. Abusing notation, we denote by $\varphi (M)$ the profile in $\cal M$ obtained after applying the same transformation $\varphi$ to each element in $M$; analogously, let $\varphi (s_c)$ denote the ordered sequence of endpoints obtained after applying $\varphi$ to each element in $s_c$.

\smallskip\noindent
\textbf{Stability.} For every $\varphi \in \Phi^\uparrow$, $\varphi(f(M)) = f(\varphi (M))$.
\smallskip

The last property is a continuity assumption, where the domain and the range of $f$ are endowed with the standard topologies. Under separability, this implies that changes in the collective $k$-th endpoint stay small after suitably small changes in the corresponding individual $k$-th endpoints. An equivalent formulation is to drop continuity and assume stability after replacing $\Phi^\uparrow$ with the superset $\Phi^{c\uparrow}$ of all continuous nondecreasing surjections $\varphi: X \to X$; see \citet[Prop.~5.4]{mar04}.

\smallskip\noindent
\textbf{Continuity.} The aggregation rule $f$ is continuous. 

\subsection{Main result}

Given a vector $\textbf{x}=(x_1, \ldots, x_n)$ in $\RR^n$, let $x_{(1)} \le \ldots \le x_{(k)} \le \ldots \le x_{(n)}$ be an arrangement of its components in nondecreasing order. For any $k= 1, \ldots, n$, the $k$-th \emph{order statistic} $\texttt{os}_k: \RR^n \to \RR$ is the function $\texttt{os}_k (\textbf{x}) = x_{(k)}$.

We define a class of separable aggregation rules called \emph{$p$-rules}, where $p$ is a mnemonic for `positional'. Let $\textbf{p} = \{ p_1, \ldots, p_k, \ldots, p_m\}$ be a nondecreasing sequence of $m$ (non-zero) integers with $p_m \le n$. Given $M$, the $p$-rule sets the collective $k$-th endpoint equal to the individual endpoint in the $p_k$-th position of $M^k$. That is, $$f^{\textbf{p}}(M) = \left( M^1_{(p_1)}, M^2_{(p_2)}, \ldots, M^m_{(p_m)} \right)$$
with $1 \le p_1 \le p_2 \le \ldots \le p_m \le n$. A $p$-rule designates as the $k$-th collective endpoint the $p_k$-th endpoint  among all the $k$-th individual endpoints arranged in nondecreasing order. 

The restriction $p_k \le p_{k+1}$ for $k=1, \ldots, n-1$ is needed to ensure strict consistency for any $k$. By the properties of order statistics, $M^k \ll M^{k+1}$ (in the component-wise order) implies $M^k_{(p_k)} < M^{k+1}_{(p_k)}$. The restriction $p_k \le p_{k+1}$ implies $M^{k+1}_{(p_k)} \le M^{k+1}_{(p_{k+1})}$. Combining the two inequalities, $M^k_{(p_k)} < M^{k+1}_{(p_{k+1})}$ follows. The next theorem follows from known results in aggregation theory; see Appendix~\ref{app1} for details.

\begin{thm}\label{thm1}
A separable aggregation rule for homogeneous vocabularies is strictly consistent, unanimous, anonymous, stable and continuous if and only if it is a $p$-rule.
\end{thm}

Appendix~\ref{app2} shows that the four axioms in Theorem~\ref{thm1} are independent.

\subsection{Interval aggregation}\label{sec33}

Motivated by the study of the legal standard of tolerance, \citet{mil25} considers a problem of interval aggregation. In his setup (which has inspired ours), a \emph{standard} of tolerance is the interval on the real line encompassing all those cases that are deemed acceptable. Each agent has her own standard. The question is to find an aggregation rule that inputs individual standards and outputs one collective standard; or, to put it differently, to aggregate $n$ individual intervals into one collective interval. \citet{mil25} characterizes a family of aggregation methods, aptly called \emph{endpoint rules}, for $m=2$.

Interval aggregation is a special case of vocabulary aggregation. A visual example should suffice to see why. Consider Figure~\ref{fig:five}, where the individual standard of an agent is the thick interval $(a,b)$, identified by the two endpoints $a,b$. 
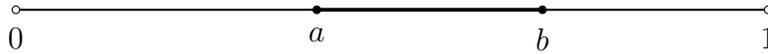
\begin{figure}[!htbp]
\centering
	\begin{tikzpicture}[scale=0.1]
		\draw[-, thick] (-0.5,30) -- (40,30);
		\draw[-, line width=0.5mm] (40,30) -- (70,30);
		\draw[-, thick] (70,30) -- (100.5,30);
		\draw[fill=white] (0,30) circle (5mm) node[below=1mm] {$0$};
		\draw[fill=black] (40,30) circle (5mm) node[below=1mm] {$a$};
		\draw[fill=black] (70,30) circle (5mm) node[below=1mm] {$b$};
		\draw[fill=white] (100,30) circle (5mm) node[below=1mm] {$1$};
		\end{tikzpicture}
	\caption{An individual interval}\label{fig:five}
\end{figure}
The same two endpoints identify a vocabulary formed by the three words $w_1=(0,a),  w_2: [a,b), w_3=[b,1)$,  interpreted as `under-standard', `standard', and `over-standard'. (These names are suggestive, but otherwise irrelevant.) Hence, the problem of aggregating $n$ individual intervals into one collective interval is equivalent to aggregating $n$ individual 3-word vocabularies into a collective 3-word vocabulary. 

Under the restriction $m=2$, \citet{mil25} defines a subclass of symmetric aggregation rules for interval aggregation. We show how to generalize them to vocabulary aggregation when $m$ is even (regardless of $n$), or when $m$ is odd and $n$ is even. For short, we refer to this as the \emph{parity condition}.

A $p$-rule is \emph{symmetric} if $p_k + p_{m-k+1} = n + 1$, for all $k=1, \ldots, m$. The symmetry refers to the property that the twin positions $p_k$ and $p_{m-k+1}$ are spaced $k-1$ items away from the first and from the last endpoint, respectively. If one reverses the direction of the line, the reversed positions would match the original positions under the initial order. 

An example may be useful. Suppose $m=3$ endpoints and $n=5$ agents over the interval $(0,11)$, with the individual endpoints arranged in the matrix
$$M^< = \begin{bmatrix}
1 & 4 & 9\\
2 & 6 & 7\\
3 & 4 & 10\\
4 & 5 & 6\\
5 & 6 & 8
\end{bmatrix}$$
Consider the $p$-rule $(2,3,4)$, where $p_1 + p_3 = 2 + 4 = 6 = n+1$. Under the linear order $<$, this $p$-rule selects the three endpoints $M^1_{(2)} = 2$, $M^2_{(3)} = 5$, and $M^3_{(4)} = 9$. That is, the first collective endpoint is the second lowest among the first individual endpoints, while the last collective endpoint is the fourth lowest among the last individual endpoints.

Reverse the order throughout from $<$ to $>$, so that what used to read ``higher'' (to the right) is now ``lower'' (to the left). The individual endpoints can be arranged in a new (order-reversed) matrix
$$M^> = \begin{bmatrix}
9 & 4 & 1\\
7 & 6 & 2\\
10 & 4 & 3\\
6 & 5 & 4\\
8 & 6 & 5
\end{bmatrix}$$
Under the new order $>$, the same $p$-rule $(2,3,4)$ with $p_1 + p_3 = n+1$ selects the first collective endpoint as the second \emph{highest} (and fourth lowest) among the first individual endpoints, while the last collective endpoint is the fourth \emph{highest} (and second lowest) among the last individual endpoints. The collective endpoints are $(9, 5, 2)$. Because the aggregation rule is symmetric, the selection of collective endpoints remains unchanged under the original order $<$ or the reversed order $>$.

In particular, for $m=2$, the symmetric rule that sets $p_1 = \lfloor \frac{n+1}{2} \rfloor$ and $p_2 = \lceil \frac{n+1}{2} \rceil$ recovers the median rule suggested by \citet{mil25} and others. More generally, under the parity condition, we define the \emph{median rule} by setting $p_k = \lfloor \frac{n+1}{2} \rfloor$ for $k = 1, \ldots, \lfloor \frac{m+1}{2} \rfloor$ and $p_k = \lceil \frac{n+1}{2} \rceil$ for $k = \lceil \frac{m+1}{2} \rceil, \ldots, m$.

To characterize the symmetric endpoint rules, we need a stronger notion of stability that replaces $\Phi^\uparrow$ with the superset $\Phi^\updownarrow$ of all strictly monotonic bijections $\varphi: X \to X$. While stability expects the aggregation rule to be order-preserving, strong stability adds the requirement that the rule is invariant to order inversions.\footnote{The convention about left endpoints is independent of strong stability: it is used after deriving the collective endpoints only to attribute each of them to one word.}

\smallskip\noindent
\textbf{Strong stability.} For every $\varphi \in \Phi^\updownarrow$, $\varphi(f(M)) = f(\varphi (M))$.
\smallskip

The next result is a close analog of Theorem~\ref{thm1}, where strong stability replaces stability. 
\begin{thm}\label{thm2}
Assume the parity condition. A separable aggregation rule for homogeneous vocabularies is strictly consistent, unanimous, anonymous, strongly stable and continuous if and only if it is a symmetric $p$-rule.
\end{thm}
\begin{proof}
In light of Theorem~\ref{thm1}, it suffices to prove that for a $p$-rule strong stability implies symmetry. The opposite direction is trivial. Let $\Phi^\downarrow$ be the set of all decreasing bijections $\varphi: X \to X$. Clearly, $\Phi^\updownarrow = \Phi^\uparrow \cup \Phi^\downarrow$ and $\Phi^\uparrow \cap \Phi^\downarrow = \emptyset$. Applying $\varphi$ in $\Phi^\downarrow$ to any ordered sequence of $m$ elements reverses its order, placing its $k$-th element in the $(m-k+1)$-th position. 

Assume $\varphi (f (M)) = f (\varphi (M))$ for some $\varphi$ in $\Phi^\downarrow$. Given a $p$-rule $f$, $\varphi (f (M))$ sends the $k$-th collective endpoint $f^k (M^k)$ in the $(m-k+1)$-th position as $\varphi (f^k (M^k))$ and has the $\varphi$-transform of the lower $p_{(k)}$-th individual endpoint from $M^k$ as its $(m-k+1)$-th component. 

On the other hand, because $\varphi (M)$ swaps the $k$-th column with the $(m-k+1)$-th column of $M$, the $(m-k+1)$-th component of $f (\varphi (M))$ picks the lower $p_{(m-k+1)}$-th individual endpoint from $\varphi (M^{k})$. Removing the order reversal, this is the same as the higher $p_{(m-k+1)}$-th endpoint from $M^k$. Then, because there are $n$ individual endpoints, the higher $p_{(m-k+1)}$-th individual endpoint is equal to the lower $\left( n+1-p_{(m-k+1)}\right)$-th endpoint. Hence, strong stability requires $p_{(k)} = n+1-p_{(m-k+1)}$, or $p_{(k)} + p_{(m-k+1)} = n +1$.
\end{proof}

\section{Heterogeneous vocabularies}\label{sec4}

When vocabularies are homogeneous, all agents share the same active words, possibly with different extents. In general, however, agents may disregard some words and attribute them empty extent.

We say that vocabularies are heterogeneous when agents may disagree not only on the extent, but also on which words are active. Assume that each individual vocabulary has at least two active words.\footnote{A one-word vocabulary covers the whole of $X$: it is uninformative, and hence irrelevant.} Note that, by the underlying linear order, all agents  agree on the order of words, regardless of which they actively use; for instance, a professor who never uses the grade $D$ agrees that it precedes $F$ and follows $C$.

The example in Figure~\ref{fig-6} depicts a vocabulary using three words out of a vocabulary of four: $w_1 \prec w_2 \prec w_3 \prec w_4$. We say that the vocabulary is \emph{partial} because the word $w_3$ is missing or, more precisely, its extent is empty. 
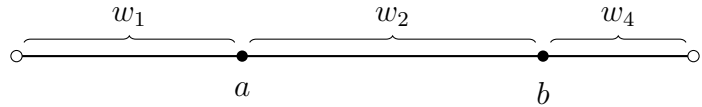
\begin{figure}[!ht]
\vspace{2mm}
\centering
\begin{tikzpicture}
	\draw[>={Latex[width=4mm]}, thick] (0, 0) -- (9, 0);
	\draw[fill=white] (0,0) circle (0.75mm);
	\draw[decorate, decoration={brace}, yshift=1ex]  (0.1,0) -- node[above=0.4ex] {$w_1$}  (2.9,0);
	\fill[black] (3,0) circle (0.75mm) node[below=2mm] {$a$};
	\draw[decorate, decoration={brace}, yshift=1ex]  (3.1,0) -- node[above=0.4ex] {$w_2$}  (6.9,0);
	\fill[black] (7,0) circle (0.75mm) node[below=2mm] {$b$};
	\draw[decorate, decoration={brace}, yshift=1ex]  (7.1,0) -- node[above=0.4ex] {$w_4$}  (8.9,0);
	\draw[fill=white] (9,0) circle (0.75mm);
\end{tikzpicture}
\caption{A partial vocabulary}\label{fig-6}
\end{figure}

For a partial vocabulary, the implicit correspondence between words and endpoints seems to break down: a mere list of the two endpoints $\{a,b\}$ identifies the three intervals $(0,a)$, $(a,b)$ and $(b,1)$, but cannot pinpoint how they associate with three of the four words in the full vocabulary. Nonetheless, we can restore the duality between words and endpoints using the common order structure. 

Intuitively, it is enough to associate the missing word $w_3$ with the (empty) interval $(b,b)$ identified by two copies of the endpoint $b$. Thus, the list of endpoints associated with the partial vocabulary in Figure~\ref{fig-6} is the (ordered) multiset $\{a, b, b \}$. Analogously, if the partial vocabulary has the nontrivial words $\{ w_1, w_3, w_4 \}$, the list of its endpoints is $\{a, a, b \}$. If the first or last word is missing, we associate it respectively with the interval $(0,0)$ or $[1,1)$; that is, we add the boundary points of the domain as possible endpoints. Similarly, the partial 2-word vocabulary $\{w_2, w_4 \}$ is uniquely associated with the (ordered) 3-element multiset $(0, a,a)$, where $a$ is the endpoint separating $w_2$ and $w_4$. See Figure~\ref{fig-7}.
\begin{figure}[!ht]
\vspace{2mm}
\centering
\begin{tikzpicture}
	\draw[>={Latex[width=4mm]}, thick] (0, 0) -- (9, 0);
	\draw[fill=white] (0,0) circle (0.75mm) node[below=2mm] {$0$};
	\draw[decorate, decoration={brace}, yshift=1ex]  (0.1,0) -- node[above=0.4ex] {$w_2$}  (3.9,0);
	\fill[black] (4,0) circle (0.75mm) node[below=2mm] {$a$};
	\draw (4,0) node[below=7mm] {$a$};
	\draw[decorate, decoration={brace}, yshift=1ex]  (4.1,0) -- node[above=0.4ex] {$w_4$}  (8.9,0);
	\draw[fill=white] (9,0) circle (0.75mm);
\end{tikzpicture}
\caption{A partial 2-word vocabulary and its three endpoints}\label{fig-7}
\end{figure}
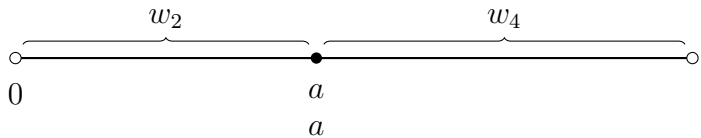

Formally, let $m+1$ be the number of available words and assume a partial vocabulary with $\ell +1 \le m+1$ words. The duality between words and endpoints uniquely associates to each $\ell$-word vocabulary a different multiset of $m$ endpoints, chosen from the closure of the set $X$. Endpoints are \emph{proper} if they belong to $X$. 

With $\ell = m$ endpoints, all words are active and the multiset has $m$ proper endpoints with no multiple instances. When $\ell < m$, some words  are inactive and we increase the multiplicity of their endpoints to account for the inactive words. Because any vocabulary $V$ has at least two words, there is at least a proper endpoint $a \in (0,1)$. For any inactive word $w_i$ before the first active word, we set $w_i = (\inf X, \inf X)$ and increase the multiplicity of this improper endpoint by~1. For any inactive word $w_i$ between two active words over adjacent intervals $(a,b)$ and $(b,c)$, we set $w_i = (b,b)$ and increase the multiplicity of the endpoint $b$ by~1. Finally, for any inactive word $w_i$ after the last active word, we set $w_i = (\sup X, \sup X)$ and increase the multiplicity of this improper endpoint by~1.

This procedure restores the duality. Every vocabulary with $m+1$ words (partial or not) is associated with a (unique) ordered sequence of exactly $m$ endpoints $s^1 \le s^2 \le \ldots \le s^m$, chosen from $X \cup \{ \inf X, \sup X \}$. With respect to the case of homogeneous vocabularies, there are three changes: 1) $\inf X$ and $\sup X$ can be endpoints; 2) the endpoints form a multiset, whose multiplicities are associated with the number of inactive words; 3) when arranged in an ordered sequence, inequalities between endpoints are weak.

\subsection{Results}

When vocabularies are partial, the duality between words and endpoints is preserved after expanding the set of endpoints from $X$ to $\overbar{X} = X \cup \{ \inf X, \sup X \}$ and coding missing intervals by means of multisets of endpoints. Because endpoints can have multiplicity larger than one, the ordered sequence from a multiset satisfies $s_1 \le s_2 \le \ldots \le s_m$ with weak inequalities. 

Let $\overbar{\cal M}$ be the set of profiles of individual endpoints from $\overbar{X}$. An aggregation rule $f: \overbar{\cal M} \to X^m$ maps each profile $M$ of individual endpoints from $\overbar{X}$ into a vector $s_c = (s_c^1, \ldots, s_c^k, \ldots, s_c^m)$ of collective endpoints in $\overbar{X}^m$. We keep our focus on aggregation rules that are separable and consistent.

Clearly, partial individual vocabularies may aggregate into a partial collective vocabulary. This is also a consequence of unanimity: if all agents agree on the same partial vocabulary, so must do the collective vocabulary. Therefore, we relax the requirement of strict consistency used for homogeneous vocabularies to (weak) consistency.

\smallskip\noindent
\textbf{Consistency.} For any $M$, $f^1(M^1) \le \ldots \le f^k(M^k) \le \ldots \le f^m(M^m)$.
\smallskip

Theorem~\ref{thm1} straightforwardly extends to heterogeneous vocabularies. The analog holds for Theorem~\ref{thm2} under the parity condition, but we do not pursue it here.

\begin{thm}\label{thm3}
A separable aggregation rule for vocabularies is consistent, unanimous, anonymous, stable and continuous if and only if it is a $p$-rule.
\end{thm}

\subsection{Majoritarian consensus}

An example illustrating vocabulary aggregation over partial vocabularies may be useful.

\begin{example}\label{ex-2}
Consider a vocabulary with eight words $w_1 \prec w_2 \prec \ldots \prec w_8$. There are three agents with partial vocabularies, depicted in the top three lines of Figure~\ref{fig-8}. The active words in the first 3-item vocabulary are $w_1$ on $(0,c)$, $w_2$ on $[c,h)$ and $w_3$ on $[h,1)$. The active words in the second 4-item vocabulary are $w_2$ on $(0,b)$, $w_4$ on $[b,d)$, $w_6$ on $[d,g)$, and $w_8$ on $[g,1)$. The active words in the third 5-item vocabulary are $w_1$ on $(0,a)$, $w_3$ on $[a,e)$, $w_5$ on $[e,f)$, $w_6 $ on $[f,1)$, and $w_7$ on $[i,1)$. We depict all the endpoints in the fourth row: the alphabetic order of their names matches their ordered positions on the line.
 
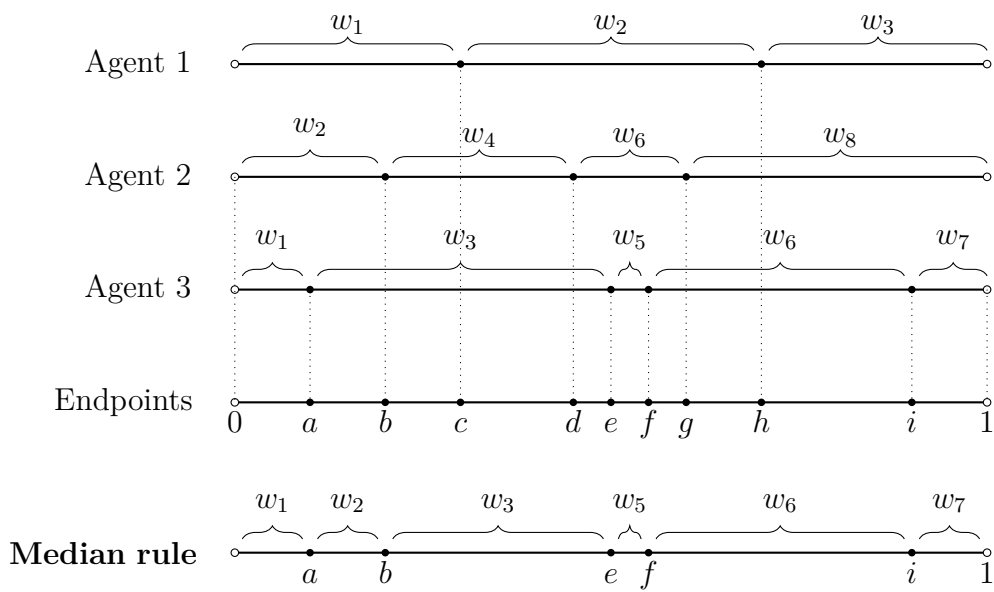
\begin{figure}[!ht]
\centering
\begin{tikzpicture}[scale=0.1]
\draw[-, thick] (-0.5, 0) -- (100.5, 0);
\draw[-, thick] (-0.5, 15) -- (100.5, 15);
\draw[-, thick] (-0.5, 30) -- (100.5, 30);
\draw[-, thick] (-0.5, -15) -- (100.5, -15);		
\draw[-, thick] (-0.5, -35) -- (100.5, -35);

\node[text width=2cm] at (-10,30) {Agent $1$};
\node[text width=2cm] at (-10,15) {Agent $2$};
\node[text width=2cm] at (-10,00) {Agent $3$};
\node[text width=2.8cm] at (-10,-15) {Endpoints};
\node[text width=4cm] at (-10,-35) {\textbf{Median rule}};

\draw[decorate, decoration={brace, amplitude=5pt}, yshift=10ex] (1,30) -- (29,30) node[midway, above=0.4ex] {$w_1$};
\draw[decorate, decoration={brace, amplitude=5pt}, yshift=10ex] (31,30) -- (69,30) node[midway, above=0.4ex] {$w_2$};
\draw[decorate, decoration={brace, amplitude=5pt}, yshift=10ex] (71,30) -- (100,30) node[midway, above=0.4ex] {$w_3$};

\draw[decorate, decoration={brace, amplitude=5pt}, yshift=10ex] (1,15) -- (19,15) node[midway, above=6pt] {$w_2$};
\draw[decorate, decoration={brace, amplitude=5pt}, yshift=10ex] (21,15) -- (44,15) node[midway, above=0.4ex] {$w_4$};
\draw[decorate, decoration={brace, amplitude=5pt}, yshift=10ex] (46,15) -- (59,15) node[midway, above=0.4ex] {$w_6$};
\draw[decorate, decoration={brace, amplitude=5pt}, yshift=10ex] (61,15) -- (100,15) node[midway, above=0.4ex] {$w_8$};

\draw[decorate, decoration={brace, amplitude=5pt}, yshift=10ex] (1,0) -- (9,0) node[midway, above=6pt] {$w_1$};
\draw[decorate, decoration={brace, amplitude=5pt}, yshift=10ex] (11,0) -- (49,0) node[midway, above=6pt] {$w_3$};
\draw[decorate, decoration={brace, amplitude=5pt}, yshift=10ex] (51,0) -- (54,0) node[midway, above=6pt] {$w_5$};
\draw[decorate, decoration={brace, amplitude=5pt}, yshift=10ex] (56,0) -- (89,0) node[midway, above=6pt] {$w_6$};
\draw[decorate, decoration={brace, amplitude=5pt}, yshift=10ex] (91,0) -- (100,0) node[midway, above=6pt] {$w_7$};

\draw[fill=white] (0, 0) circle (5mm);
\draw[fill=white] (0, 15) circle (5mm);
\draw[fill=white] (0, 30) circle (5mm);
\fill (10, 0) circle (5mm);
\fill (20,15) circle (5mm);
\fill (30,30) circle (5mm);
\fill (45,15) circle (5mm);
\fill (50,0) circle (5mm);
\fill (55,0) circle (5mm);
\fill (60,15) circle (5mm);
\fill (70,30) circle (5mm);
\fill (90,0) circle (5mm);
\draw[fill=white] (100,0) circle (5mm);
\draw[fill=white] (100,30) circle (5mm);
\draw[fill=white] (100,15) circle (5mm);

\draw[fill=white] (0,-15) circle (5mm) node[below=4mm,anchor=base] {$0$};
\fill (10,-15) circle (5mm) node[below=4mm,anchor=base] {$a$};
\fill (20,-15) circle (5mm) node[below=4mm,anchor=base] {$b$};
\fill (30,-15) circle (5mm) node[below=4mm,anchor=base] {$c$};
\fill (45,-15) circle (5mm) node[below=4mm,anchor=base] {$d$};
\fill (50,-15) circle (5mm) node[below=4mm,anchor=base] {$e$};
\fill (55,-15) circle (5mm) node[below=4mm,anchor=base] {$f$};
\fill (60,-15) circle (5mm) node[below=4mm,anchor=base] {$g$};
\fill (70,-15) circle (5mm) node[below=4mm,anchor=base] {$h$};
\fill (90,-15) circle (5mm) node[below=4mm,anchor=base] {$i$};
\draw[fill=white] (100,-15) circle (5mm) node[below=4mm,anchor=base] {$1$};

\draw[dotted] (0,15) -- (0,-15) node[below=4mm] {};
\draw[dotted] (10,0) -- (10,-15) node[below=4mm] {};
\draw[dotted] (20,15) -- (20,-15) node[below=4mm] {};
\draw[dotted] (30,30) -- (30,-15) node[below=4mm] {};
\draw[dotted] (45,15) -- (45,-15) node[below=4mm] {};
\draw[dotted] (50,0) -- (50,-15) node[below=4mm] {};
\draw[dotted] (55,0) -- (55,-15) node[below=4mm] {};
\draw[dotted] (60,15) -- (60,-15) node[below=4mm] {};
\draw[dotted] (70,30) -- (70,-15) node[below=4mm] {};
\draw[dotted] (90,0) -- (90,-15) node[below=4mm] {};
\draw[dotted] (100,0) -- (100,-15) node[below=4mm] {};
    
\draw[decorate, decoration={brace, amplitude=5pt}, yshift=10ex] (1,-35) -- (9,-35) node[midway, above=6pt] {$w_1$};
\draw[decorate, decoration={brace, amplitude=5pt}, yshift=10ex] (11,-35) -- (19,-35) node[midway, above=6pt] {$w_2$};
\draw[decorate, decoration={brace, amplitude=5pt}, yshift=10ex] (21,-35) -- (49,-35) node[midway, above=6pt] {$w_3$};
\draw[decorate, decoration={brace, amplitude=5pt}, yshift=10ex] (51,-35) -- (54,-35) node[midway, above=6pt] {$w_5$};
\draw[decorate, decoration={brace, amplitude=5pt}, yshift=10ex] (56,-35) -- (89,-35) node[midway, above=6pt] {$w_6$};
\draw[decorate, decoration={brace, amplitude=5pt}, yshift=10ex] (91,-35) -- (99,-35) node[midway, above=6pt] {$w_7$};

\draw[fill=white] (0, -35) circle (5mm);
\fill (10,-35) circle (5mm) node[below=4mm,anchor=base] {$a$};
\fill (20,-35) circle (5mm) node[below=4mm,anchor=base] {$b$};
\fill (50,-35) circle (5mm) node[below=4mm,anchor=base] {$e$};
\fill (55,-35) circle (5mm) node[below=4mm,anchor=base] {$f$};
\fill (90,-35) circle (5mm) node[below=4mm,anchor=base] {$i$};
\draw[fill=white] (100,-35) circle (5mm) node[below=4mm,anchor=base] {$1$};

\end{tikzpicture}
\caption{Aggregation of heterogeneous vocabularies by the median rule}\label{fig-8}
\end{figure}
		
If we juxtapose the ordered sequence of elements from the three corresponding multisets of endpoints, we obtain the matrix-like arrangement
$$M = \begin{bmatrix}
c & h & 1 & 1 & 1 & 1 & 1 \\
0 & b & b & d & d & g & g \\
a & a & e & e & f & i & 1 \\
\end{bmatrix}$$
Applying the median rule defined in Section~\ref{sec33}, we obtain the collective multiset
$$s_c = \begin{bmatrix}
a & b & e & e & f & i & 1 \\
\end{bmatrix}$$
which (uniquely) corresponds to a partial vocabulary having six active words: $w_1$ on $(0,a)$, $w_2$ on $[a,b)$, $w_3$ on $[b,e)$, $w_5$ on $[e,f)$, $w_6$ on $ [f,i)$, and $w_7$ on $[i,1)$. This is depicted in the last line of Figure~\ref{fig-8}. 		
\end{example}
	
We make two observations. On one hand, the collective vocabulary lacks the words $w_4$ and $w_8$, which are present only in Agent $2$'s vocabulary. On the other hand, the collective vocabulary includes the words $w_5$ and $w_7$, which are present only in Agent $3$'s vocabulary. 

This leads to inquire about active words in the collective vocabulary. Because an active word has nonempty extent, we consider two properties. The first one requires that the aggregation rule keeps active any word actively used by a strict majority of agents, regardless of whether they apply it to different extents. Let $N^k = \{ i \in N: s^i_k < s^i_{k+1}\}$ be the set of agents with individual vocabularies where word $k$ is active.

\smallskip\noindent
\textbf{Majoritarian words:} If $|N^k| \geq \frac{n+1}{2}$, then $f^{k}(M^k) < f^{k+1}(M^{k+1})$.
\smallskip

\begin{prop}
Any $p$-rule may fail to preserve majoritarian words.     
\end{prop}
\begin{proof}
Assume three agents, with respective vocabularies 
\begin{align*} 
V_1 & = \{ w_1 \mbox{ is } (0,a); w_2 \mbox{ is } (a,1) \} \\
V_2 & = \{ w_1 \mbox{ is } (0,a); w_3 \mbox{ is } (a,1) \} \\
V_3 & = \{ w_2 \mbox{ is } (0,a); w_3 \mbox{ is } (a,1) \}
\end{align*}
as shown in Figure~\ref{fig-9}. 
\begin{figure}[!ht]
\centering
\begin{tikzpicture}[scale=0.1]
\draw[-, thick] (-0.5, 0) -- (90.5, 0);
\draw[-, thick] (-0.5, 15) -- (90.5, 15);
\draw[-, thick] (-0.5, 30) -- (90.5, 30);

\node[text width=2cm] at (-10,30) {Agent $1$};
\node[text width=2cm] at (-10,15) {Agent $2$};
\node[text width=2cm] at (-10,00) {Agent $3$};

\draw[decorate, decoration={brace, amplitude=5pt}, yshift=10ex] (1,30) -- (44,30) node[midway, above=0.4ex] {$w_1$};
\draw[decorate, decoration={brace, amplitude=5pt}, yshift=10ex] (46,30) -- (89,30) node[midway, above=0.4ex] {$w_2$};

\draw[decorate, decoration={brace, amplitude=5pt}, yshift=10ex] (1,15) -- (44,15) node[midway, above=6pt] {$w_1$};
\draw[decorate, decoration={brace, amplitude=5pt}, yshift=10ex] (46,15) -- (89,15) node[midway, above=0.4ex] {$w_3$};

\draw[decorate, decoration={brace, amplitude=5pt}, yshift=10ex] (1,0) -- (44,0) node[midway, above=6pt] {$w_2$};
\draw[decorate, decoration={brace, amplitude=5pt}, yshift=10ex] (46,0) -- (89,0) node[midway, above=6pt] {$w_3$};

\draw[fill=white] (0, 0) circle (5mm);
\draw[fill=white] (0, 15) circle (5mm);
\draw[fill=white] (0, 30) circle (5mm);
\fill (45, 0) circle (5mm) node[below=4mm,anchor=base] {$a$};
\fill (45,15) circle (5mm) node[below=4mm,anchor=base] {$a$};
\fill (45,30) circle (5mm) node[below=4mm,anchor=base] {$a$};
\fill (90,0) circle (5mm);
\draw[fill=white] (90,0) circle (5mm);
\draw[fill=white] (90,15) circle (5mm);
\draw[fill=white] (90,30) circle (5mm);
 
\end{tikzpicture}
\caption{Three heterogeneous individual vocabularies}\label{fig-9}
\end{figure}
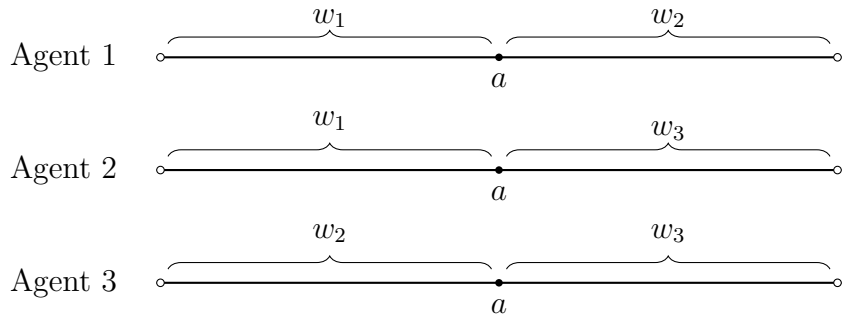
Each individual vocabulary has two active words, and each word is used by at least two agents. On the other hand, the set of all individual endpoints is a singleton. Hence, because any $p$-rule selects its collective endpoints from the set of individual endpoints, there cannot be three active words with nonempty extent.
\end{proof}

Contrary to $p$-rules, the mean rule preserves majoritarian words. Indeed, the mean rule satisfies the stronger property of preserving any word for which at least one agent has nonempty extent. Indeed, assume that agent $i$ actively uses word $w_k$; that is, $s^k_i < s^{k+1}_i$. (A similar argument holds for $0 < s^1_i$ or $s^m_i<1$.) Because $s^k_j \leq s^{k+1}_j$ for any other agent $j$, we have $f^k_{ave}(\textbf{s}) < f^{k+1}_{ave}(\textbf{s})$ and the collective extent for word $k$ is not empty. In the example above, for instance, the mean rule yields the two endpoints $\frac{2a}{3}$ and $\frac{1+2a}{3}$.

A second stronger property requires that the aggregation rule preserves the extent of an active word when it is shared by a strict majority of agents. Let $N^k (a,b) = \{ i \in N: s^i_k \le a < b \le s^i_{k+1}\}$ be the set of agents whose individual vocabularies agree that the extent $(a,b)$ is included in word $k$.

\smallskip\noindent
\textbf{Majoritarian extents:} If $|N^k (a,b)| \geq \frac{n+1}{2}$, then $f^{k}(M^k) \leq a < b \leq f^{k+1}(M^{k+1})$.
\smallskip

The only $p$-rule that respects majoritarian extents is the median rule, provided that the number of agents $n$ is odd.
	
\begin{prop}\label{thm6}
A $p$-rule preserves majoritarian extents if and only if $n$ is odd and $p_k = \frac{n+1}{2}$ for all $k$.     
\end{prop}
\begin{proof}
We prove only necessity because sufficiency is trivial. An interval $(a,b)$ is approved by a strict majority when there are (at least) $\lceil \frac{n+1}{2} \rceil$ agents whose individual vocabularies have endpoints $s^k_i < a < b < s^{k+1}_i$ for the same $k$. We need to ensure that the collective endpoints $s^k_c$ and $s^{k+1}_c$ satisfy $s^k_c < a < b < s^{k+1}_i$.

Because any $p$-rule selects the $k$-th collective endpoint $s^k_c$ from the union of the $k$-th individual endpoints $\{ s^k_i \}$, we need $p_k \ge \lceil \frac{n+1}{2} \rceil$ to ensure that at least $\lceil \frac{n+1}{2} \rceil$ agents approve $a$. Similarly, we must have $p_{k+1} \le n - \lceil \frac{n+1}{2} \rceil + 1 = \lfloor \frac{n+1}{2} \rfloor$ so that at least $\lceil \frac{n+1}{2} \rceil$ approve $b$. Chaining the inequalities, we find
$$\left\lceil \frac{n+1}{2} \right\rceil \le p_k \le p_{k+1} \le \left\lfloor \frac{n+1}{2} \right\rfloor$$
and the conclusion follows because $\lceil \frac{n+1}{2} \rceil \le \lfloor \frac{n+1}{2} \rfloor$ holds only for $n$ odd.
\end{proof}

Proposition~\ref{thm6} provides a dim view of the possibility of respecting strict majorities. However, if the property is relaxed to a weak majority of at least 50\% of the agents, a straightforward modification of the proof shows that a $p$-rule preserves weakly majoritarian extents if and only if $\frac{n}{2} \le p_k \le p_{k+1} \le \frac{n+1}{2}$ for all $k$. In particular, when the parity condition holds, the median rule for $m$ endpoints respects weakly majoritarian extents. 

\section{Incomplete vocabularies}\label{sec5}

Until now, we have assumed that each agent reports the partition that identifies their individual vocabulary. \citet{far11} argue that this may require a ``significant deliberation effort'', and suggest the possibility to query agents only about specific examples. While this does not generally suffice to derive a collective vocabulary for the whole range of $X$, it may provide enough information for some purposes.

Suppose that all agents are shown the same finite set of \emph{exemplars} $E=\{ e_1, \ldots, e_r \}$ from $X$. Given an $(m+1)$-item vocabulary, an agent attaches each point in $E$ to one of the $m+1$ words. (It is reasonable to expect $r > m+1$ so that agents, facing more exemplars than words, have a chance to use them all; however, this is not necessary.) 

For example, assume a vocabulary with six words $w_1 \prec w_2 \prec \ldots \prec  w_6$. The first line in Figure~\ref{fig-10} depicts five exemplars $e_1 < e_2 < \ldots < e_5$ and their attached words.
\begin{figure}[!ht]
\vspace{2mm}
\centering
\begin{tikzpicture}
	\draw[>={Latex[width=4mm]}] (0, 2) -- (12, 2);
	\draw[fill=white] (0,2) circle (0.75mm);
	\fill[black] (2,2) circle (0.75mm) node[above=2mm] {$w_2$};
	\draw (2,2) node[below=2mm] {$e_1$};
	\fill[black] (4,2) circle (0.75mm) node[above=2mm] {$w_3$};
	\draw (4,2) node[below=2mm] {$e_2$};
	\fill[black] (6,2) circle (0.75mm) node[above=2mm] {$w_3$};
	\draw (6,2) node[below=2mm] {$e_3$};
	\fill[black] (8,2) circle (0.75mm) node[above=2mm] {$w_5$};
	\draw (8,2) node[below=2mm] {$e_4$};
	\fill[black] (10,2) circle (0.75mm) node[above=2mm] {$w_6$};
	\draw (10,2) node[below=2mm] {$e_5$};
	\draw[fill=white] (12,2) circle (0.75mm);
		
	\draw[>={Latex[width=4mm]}] (0, 0) -- (12, 0);
	\draw[fill=white] (0,0) circle (0.75mm);
	\fill[black] (2,0) circle (0.75mm) node[above=2mm] {$w_2$};
	\draw (2,0) node[below=2mm] {$e_1$};
	\fill[black] (4,0) circle (0.75mm);
	\draw[decorate, decoration={brace}, yshift=1ex]  (4.01,0) -- node[above=0.4ex] {$w_3$}  (5.99,0);
	\draw (4,0) node[below=2mm] {$e_2$};
	\fill[black] (6,0) circle (0.75mm);
	\draw (6,0) node[below=2mm] {$e_3$};
	\fill[black] (8,0) circle (0.75mm) node[above=2mm] {$w_5$};
	\draw (8,0) node[below=2mm] {$e_4$};
	\fill[black] (10,0) circle (0.75mm);
	\draw[decorate, decoration={brace}, yshift=1ex]  (10.01,0) -- node[above=0.4ex] {$w_6$}  (11.99,0);
	\draw (10,0) node[below=2mm] {$e_5$};
	\draw[fill=white] (12,0) circle (0.75mm);
	\draw[-, line width=0.6mm] (4,0) -- (6,0);
	\draw[-, line width=0.6mm] (10.05,0) -- (11.95,0);
\end{tikzpicture}
\caption{Five exemplars and their induced vocabulary}\label{fig-10}
\end{figure}
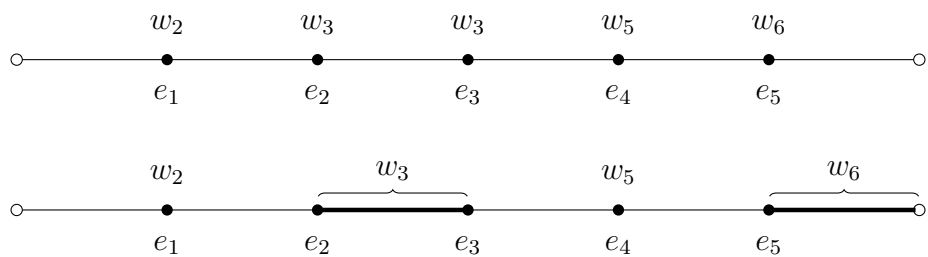

We leverage the assumption that $X$ is a linear order to draw some conclusions. If any two exemplars $e < e^\prime$ are attached to the same word, so must be the interval $[e,e^\prime]$. In the example, the word $w_3$ extends from the two exemplars $e_2 < e_3$ to the interval $[e_2, e_3]$. In general, if we denote by $E^k$ the finite set of exemplars associated with a word $w_k$, the interval $[\min E^k, \max E^k]$ is assigned to $w_k$. 

By a similar argument, if the first (or last) exemplar is attached to the first (or last) word, one may extend the word to the respective endpoint. In the example, the word $w_6$ extends from $e_5$ to the interval $[e_5, \sup X)$. These natural extensions generate the incomplete vocabulary depicted in the second line of Figure~\ref{fig-10}, which we call the \emph{induced vocabulary}.

This section makes three arguments. First, we show how to extend the duality between words and endpoints to incomplete vocabularies. Second, we exploit this duality to define aggregation rules over incomplete vocabularies; although one cannot recover a complete collective vocabulary from the aggregation of several incomplete vocabularies, it is possible to reach partial conclusions. Third, we exemplify the argument using a $p$-rule to aggregate individual and incomplete vocabularies. We focus on the case of homogeneous incomplete vocabularies, where all agents are given the same finite set of exemplars to label and everyone labels all the exemplars. 

We begin with duality. When a vocabulary is complete, endpoints separate words. When a vocabulary is incomplete, singletons or intervals attached to a word are separated by gaps. Return to the second line in Figure~\ref{fig-10}: the induced vocabulary is a mixture of four partial words (attached to two singletons and two intervals, in bold), with gaps in between. 

The key observation is that the alternation between nonempty extents and gaps implies that an induced vocabulary with $m+1$ words must have $m$ gaps (possibly, with multiplicity greater than one). The duality is restored by treating gaps as if they were endpoints. Every incomplete vocabulary with $m+1$ words is associated with a (unique) ordered sequence of exactly $m$ gaps $g^1 \preceq g^2 \preceq \ldots \preceq g^m$, where $\prec$ is the natural order between consecutive intervals. 

Returning to the second line in Figure~\ref{fig-10}, the active (partial) words are $w_2 \prec w_3 \prec w_5 \prec w_6$. Arrange the five gaps delimiting them in increasing order: the gap $g_k$ separates the word $w_k$ from $w_{k+1}$. The five gaps are grouped $g_1 | g_2 | g_3g_4| g_5$. The two gaps $g_3,g_4$ are joined because the induced vocabulary carries insufficient information to separate them. Figure~\ref{fig-11} exhibits the duality, placing words above the line and gaps below it.
\begin{figure}[!ht]
\vspace{2mm}
\centering
\begin{tikzpicture}
	\draw[>={Latex[width=4mm]}] (0, 0) -- (12, 0);
	\draw[fill=white] (0,0) circle (0.75mm);
	\fill[black] (2,0) circle (0.75mm) node[above=2mm] {$w_2$};
	\draw[decorate, decoration={brace,mirror}, yshift=-1ex]  (0.01,0) -- node[below=0.4ex] {$g_1$}  (1.99,0);
	\fill[black] (4,0) circle (0.75mm);
	\draw[decorate, decoration={brace}, yshift=1ex]  (4.01,0) -- node[above=0.4ex] {$w_3$}  (5.99,0);
	\draw[decorate, decoration={brace,mirror}, yshift=-1ex]  (2.01,0) -- node[below=0.4ex] {$g_2$}  (3.99,0);
	\fill[black] (6,0) circle (0.75mm);
	\draw[decorate, decoration={brace,mirror}, yshift=-1ex]  (6.01,0) -- node[below=0.4ex] {$g_3,g_4$}  (7.99,0);
	\fill[black] (8,0) circle (0.75mm) node[above=2mm] {$w_5$};
	\draw[decorate, decoration={brace,mirror}, yshift=-1ex]  (8.01,0) -- node[below=0.4ex] {$g_5$}  (9.99,0);
	\fill[black] (10,0) circle (0.75mm);
	\draw[decorate, decoration={brace}, yshift=1ex]  (10.01,0) -- node[above=0.4ex] {$w_6$}  (11.99,0);
	\draw[fill=white] (12,0) circle (0.75mm) (12,0);
	\draw[-, line width=0.6mm] (4,0) -- (6,0);
	\draw[-, line width=0.6mm] (10.05,0) -- (11.95,0);
\end{tikzpicture}
\caption{Duality between extents and gaps}\label{fig-11}
\end{figure}
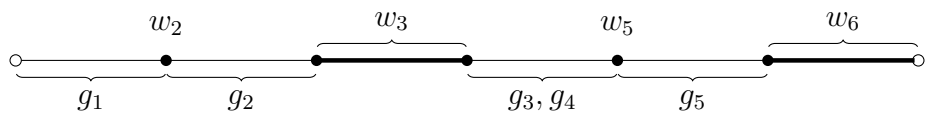

Consider aggregation rules. Under the word-gap duality, each individual (partial) vocabulary is associated with an ordered sequence $g^1 \preceq \ldots \preceq g^m$ of gaps. Given homogeneous vocabularies and $n$ agents, we juxtapose the individual ordered sequences of gaps into a matrix-like arrangement 
$$M = \begin{bmatrix}
g_1^1 & g_1^2 & \cdots & g_1^m\\
g_2^1 & g_2^2 & \cdots & g_2^m\\
\vdots & \vdots & \vdots & \vdots \\
g_n^1 & g_n^2 & \cdots & g_n^m\\
\end{bmatrix}$$
where row $i$ contains the ordered sequence of gaps for agent~$i$, and column~$k$ lists the $k$-th gaps for all agents. The rest of the notation is as before.

An aggregation rule $f: {\cal M} \to X^m$ maps each profile of individual gaps in $\cal M$ into a vector $g_c = (g_c^1, \ldots, g_c^k, \ldots, g_c^m)$ of collective gaps in $X^m$. Mirroring the construction above, we can state properties that characterize the aggregation of individual incomplete vocabularies. The setup is similar, but the aggregation yields collective gaps that define a collective incomplete vocabulary. 

We illustrate the aggregation procedure and the resulting collective incomplete vocabulary with an example.

\begin{example}\label{ex-3}
Suppose a vocabulary with four words. Three agents are given the same three exemplars $a < b < c$. Agent~1 uses the words $w_1$ for $a$, $w_3$ for $b$, $w_4$ for $c$. Agent~2 uses $w_2$ for $a$ and $b$, $w_3$ for $c$. Agent~3 uses $w_1$ for $a$, and $w_3$ for $b$ and $c$.

The induced vocabularies are depicted in the top three lines of Figure~\ref{fig-12}. We position the gaps for each individual (induced) vocabulary below each line, with $g^k_i$ being the $k$-th gap for agent~$i$. 
\begin{figure}[!ht]
\centering
\begin{tikzpicture}[scale=0.1]
\draw[-] (-0.5, 0) -- (100.5, 0);
\draw[-] (-0.5, 15) -- (100.5, 15);
\draw[-] (-0.5, 30) -- (100.5, 30);
\draw[-] (-0.5, -15) -- (100.5, -15);		
\draw[-] (-0.5, -30) -- (100.5, -30);

\node[text width=2cm] at (-10,30) {Agent $1$};
\node[text width=2cm] at (-10,15) {Agent $2$};
\node[text width=2cm] at (-10,00) {Agent $3$};
\node[text width=2.8cm] at (-10,-15) {Collective};
\node[text width=4cm] at (-10,-30) {\textbf{Vocabulary}};

\node[below=0.1ex] at (37.5,30) {$g_1^1,g_1^2$};
\node[below=0.1ex] at (62.5,30) {$g_1^3$};
\node[below=0.1ex] at (12.5,15) {$g_2^1$};
\node[below=0.1ex] at (62.5,15) {$g_2^2$};
\node[below=0.1ex] at (87.5,15) {$g_2^3$};
\node[below=0.1ex] at (37.5,0) {$g_3^1,g_3^2$};
\node[below=0.1ex] at (87.5,0) {$g_3^3$};
\node[below=0.1ex] at (37.5,-15) {$g_c^1,g_c^2$};
\node[below=0.1ex] at (87.5,-15) {$g_c^3$};
\draw[fill=white] (0, 0) circle (5mm);
\draw[fill=white] (0, 15) circle (5mm);
\draw[fill=white] (0, 30) circle (5mm);
\fill (25,30) circle (5mm);
\node[above] at (12.5,30) {$w_1$};
\fill (25,15) circle (5mm); 
\fill (25,0) circle (5mm);
\node[above] at (12.5,0) {$w_1$}; 
\fill (50,30) circle (5mm) node[above] {$w_3$};
\fill (50,15) circle (5mm);
\node[above] at (37.5,15) {$w_2$};
\fill (50,0) circle (5mm);
\node[above] at (62.5,0) {$w_3$};
\fill (75,30) circle (5mm); 
\node[above] at (87.5,30) {$w_4$};
\fill (75,15) circle (5mm) node[above] {$w_3$};
\fill (75,0) circle (5mm);
\draw[fill=white] (100,30) circle (5mm);
\draw[fill=white] (100,15) circle (5mm);
\draw[fill=white] (100,0) circle (5mm);

\draw[-, line width=0.6mm] (0.12,30) -- (24.88,30);
\draw[-, line width=0.6mm] (75.12,30) -- (99.88,30);
\draw[-, line width=0.6mm] (25.12,15) -- (49.88,15);
\draw[-, line width=0.6mm] (0.12,0) -- (24.88,0);
\draw[-, line width=0.6mm] (50.12,0) -- (74.88,0);

\draw[dotted] (25,30) -- (25,-30) node[below=4mm] {};
\draw[dotted] (50,30) -- (50,-30) node[below=4mm] {};
\draw[dotted] (75,30) -- (75,-30) node[below=4mm] {};

\draw[fill=white] (0, -30) circle (5mm);
\node[above] at (12.5,-30) {$w_1$};
\node[above] at (62.5,-30) {$w_3$};
\draw[-, line width=0.6mm] (0.12,-30) -- (24.88,-30);
\draw[-, line width=0.6mm] (50.12,-30) -- (74.88,-30);
\fill (25,-30) circle (5mm) node[below=4mm,anchor=base] {$a$};
\fill (50,-30) circle (5mm) node[below=4mm,anchor=base] {$b$};
\fill (75,-30) circle (5mm) node[below=4mm,anchor=base] {$c$};
\draw[fill=white] (100,-30) circle (5mm);

\end{tikzpicture}
\caption{Aggregation of incomplete vocabularies by the median rule}\label{fig-12}
\end{figure}
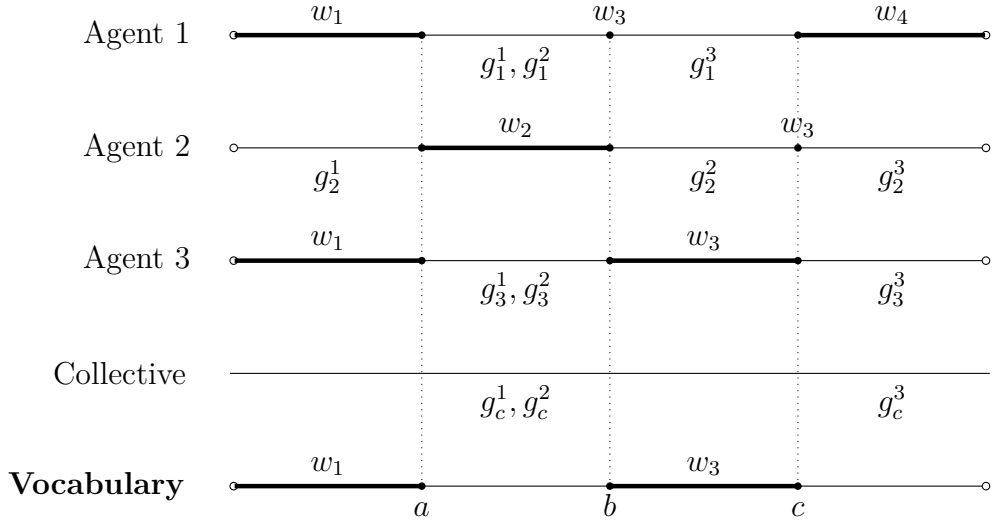

The corresponding matrix for the gaps in the individual induced vocabularies is:
$$M = \begin{bmatrix}
(a,b) & (a,b) & (b,c)\\
(\inf X, a) & (b,c) & (c,1)\\
(a,b) & (a,b) & (c,1)
\end{bmatrix}$$

We apply the median rule: for each $k$, the collective gap $g^k_c$ is the intermediate one among $\{g_1^k, g_2^k, g_3^k \}$. Proceeding by columns, the median rule selects the three collective gaps $(a,b)$, $(a,b)$, and $(c,1)$. The result is depicted in the fourth line of Figure~\ref{fig-12}. Finally, in the fifth line, we use a thick line to represent the incomplete collective vocabulary associated with the positions of the collective gaps.

The aggregation attributes the first segment $(\inf X,a)$ to $w_1$, while it remains possible that $w_1$ extends beyond that. This attribution agrees with the induced vocabularies of two agents. On the other hand, the aggregation attributes the segment $(b,c)$ to $w_3$ even if this holds only for the third agent. 
\end{example}

An important observation that we do not pursue here is that this construction is incrementally consistent. Suppose that, after having derived an incomplete collective vocabulary $V$, agents are successively shown more exemplars. If they label the new exemplars consistently with the old ones in the sense of preserving the linear order underlying them, then applying the same aggregation rule over the enlarged set of exemplars expands $V$ consistently: each collective gap in $V$ may only contract and the partial collective words from $V$ may only expand their extent.

\section{Single-peaked vocabularies}\label{sec6}

If agents have preferences over collective vocabularies, they can try and manipulate the outcome of the aggregation by falsely reporting their own individual endpoints. This section considers a class of individual preferences over (complete) vocabularies for which it is possible to characterize endpoint rules by their strategy-proofness. We build upon the special case with $m=2$ in \citet{blo10} and \citet{far11}, who discuss preferences over intervals for which the median rule is strategy-proof. 

A vocabulary with $m+1$ active words is associated to the $m$ endpoints $s^1 \le s^2 \le \ldots \le s^m$. When the vocabulary has fewer active words, some endpoints have a multiplicity higher than~1.

We start with a notion of betweenness for vocabularies. Intuitively, a vocabulary $M_1$ is between two other vocabularies $M_2$ and $M_3$ if each endpoint $s^k_1$ from $M_1$ lies between the corresponding endpoints $s^k_2$ and $s^k_3$ from $M_2$ and $M_3$, respectively. Formally, given three points $x,y,z$ in the linearly ordered set $X$, we say that $x$ lies between $y$ and $z$ and write $x \in \langle y, z \rangle$ if either $y \leq x \leq z$ or $y \geq x \geq z$. We say that $M_1$ is between $M_2$ and $M_3$ if $s^k_1 \in \langle s^k_2, s^k_3 \rangle$ for any $k = 1, \ldots m$, and we write $M_1 \in  \langle M_2, M_3 \rangle$.

A preference relation $\succeq$ over vocabularies is \emph{single-peaked} if: (i) there is a (unique) peak vocabulary $\widehat{M}$ that is strictly preferred to any $M \ne \widehat{M}$; and (ii) $\widehat{M} \succeq M_1 \succeq M_2$ for any $M_1 \in \langle \widehat{M}, M_2 \rangle$. Note that the peak vocabulary $\widehat{M}$ refers to a collection of the most preferred endpoints. Let $\cal P$ denote the class of single-peaked preference relations over vocabularies with $m$ words. 

Return to the example of Section~\ref{sec21} where Agent~1 recommends equally spaced intervals. He prefers setting the threshold for turning $F$ into $D$ at a score of 20. If he has single-peaked preferences, he dislikes moving the threshold between $F$ and $D$ away from 20 in either direction, and a similar statement holds for all the thresholds in his recommendation.

We consider strategy-proofness for the aggregation rule when the preferences of every agent are in $\cal P$. An agent with single-peaked preferences has no strict incentive to declare endpoints different from his peak. 

\smallskip\noindent
\textbf{Strategy-proofness.} For every $i$ and $\succeq_i$ in $\cal P$, $f(\widehat{M}_i, M_{-i}) \succeq_i f(M_i, M_{-i})$ for all $M \in \mathcal{M}$.
\smallskip

A key result is that strategy-proofness under single-peaked preferences implies separability. This is a generalization of \citet[Proposition 1]{mil25}; see also \citet[Theorem~2]{bor83}. The proof is in Appendix~\ref{app3}.

\begin{thm}\label{prop8}
Any strategy-proof aggregation rule is separable.
\end{thm}

Separability has two consequences. First, the aggregation rule $f(M)$ decomposes into $m$ aggregation functions $f^k(M^k)$. Second, the assumption of single-peaked preferences over vocabularies implies single-peaked preferences over each $k$-th endpoint. When each $k$-th endpoint is considered by itself, this setup is amenable to the seminal work by \citet{mou80}, who provided characterizations for peak-only strategy-proof aggregation rules. 

An aggregation rule is \emph{peak-only} if its outcome depends only on the peaks of the reported preferences. In our model, this restriction is satisfied because agents can only report their endpoints. One may derive the peak-only property from more general assumptions; see \citet[Lemma~2.1]{spr95}.  

\citet{mou80} provides two alternative formulations. We focus on the class of extended median rules. Let $\overbar{X} = X \cup \{ \inf X, \sup X \}$. Fix an ordered sequence $\textbf{q}$ of $n-1$ calibration parameters\footnote{These parameters are known as ``phantom voters'' in the social choice literature; see \citet{tho18} for a discussion about this terminology.} $q_1 \le q_2 \le \ldots \le q_{n-1}$ from $\overbar{X}$. Given a profile $\textbf{s}$ from $X^n$, the extended median rule $f_\texttt{emed}$ (with the $n-1$ parameters in $\textbf{q})$ computes 
$$f_\texttt{emed} \left( \textbf{s} , \textbf{q} \right)  = \texttt{med} \left(s_1, \ldots , s_n, q_1, \ldots, q_{n-1} \right)$$ 
The extended median recovers all the order statistics when $q_\ell \in \{ \inf X, \sup X \}$ for each $\ell$. For instance, if $q_{n-1} = \inf X$, then $f_\texttt{emed} \left( \textbf{s}, \textbf{q} \right) = s_{(1)}$; if $n$ is odd, $q_{\frac{n-1}{2}} = \inf X$ and $q_{\frac{n+1}{2}} = \sup X$, then $f_\texttt{emed} \left( \textbf{s}, \textbf{q} \right) = \texttt{med} \left( \textbf{s} \right)$. On the other hand, if some $q_\ell \not\in \{ \inf X, \sup X \}$, the extended median may return values different from the order statistics.

A corollary of \citet[Proposition~1]{mou80} is that each component $f^k(M)$ of the aggregation rule is unanimous, anonymous and strategy-proof if and only if it is an extended median rule with $n-1$ parameters.\footnote{Using instead \citet[Proposition~2]{mou80}, one may drop unanimity at the cost of assuming $n+1$ parameters; see f.i.~\citet[Claim~1]{mil25}.} The parameters may be chosen differently for each $f^k$. If we denote by $\textbf{q}^k$ the ordered sequence of $(n-1)$ parameters used for the extended median associated with $f^k$, then consistency requires $q^k_\ell \le q^{k^\prime}_\ell$ for each $\ell = 1, \ldots, n-1$ and each $k < k^\prime$. That is, the ordered sequences $\textbf{q}^k$ are nondecreasing in $k$.
 
\begin{thm}\label{thm8}
An aggregation rule $f$ is consistent, unanimous, anonymous and strategy-proof if and only if each $f^k$ is an extended median rule with $n-1$ parameters $\textbf{q}^k$ that are nondecreasing in $k$.
\end{thm}

In order to rule out those extended median rules that are not $p$-rules, one needs to make sure that all the calibration parameters are chosen in the set $\overbar{X}\setminus X = \{ \inf X, \sup X \}$. \cite{mil25} suggests a property called translation invariance, which has a technical flavor. We opt for a more intuitive requirement: if everyone strictly raises the $k$-th endpoint, then the collective endpoint should be strictly increased.

\smallskip\noindent
\textbf{Strict responsiveness.} If $M^k \ll M^{\prime k}$ then $f^k(M^k) < f^k (M^{\prime k})$. 
\smallskip

Adding this property to Theorem~\ref{thm8} yields another characterization for the $p$-rules. Unanimity can be dropped and replaced by strict responsiveness; compare Proposition~\ref{prop11} and Proposition~\ref{prop12} in the appendix. See \citet[Theorem~1]{war10} for a similar result.

\begin{thm}\label{thm9}
An aggregation rule $f$ is consistent, anonymous, strictly responsive and strategy-proof if and only if it is a $p$-rule.
\end{thm}
\begin{proof}
Sufficiency is trivial. To prove necessity, we need to show that strict responsiveness implies $q_\ell^k \in \{ \inf X, \sup X \}$ for any $p$. By contradiction, suppose that there exists some $q_\ell^k \not\in \{ \inf X, \sup X \}$. Then there are $\ell -1$ parameters before $q_\ell^k$ and $n - \ell - 1$ after. Consider a profile $M^k$ with $n - \ell$ values strictly less and $\ell$ strictly more than $q_\ell^k$. Because $q_\ell^k$ has an equal number of points (from $\textbf{q}$ and $M^k$) on either side, we have
$$\texttt{emed} (M^k; \textbf{q})  = \texttt{med} \left( M^k_1, \ldots , M^k_n, q_1, \ldots, q_{n-1} \right) = q_\ell^k$$
Let $\varepsilon_1 = \min_i \left( q_\ell^k - M^k_i , 0 \right) > 0$ and $\varepsilon_2 = \min_i \left( \sup X - q_\ell^k \right) > 0$. Choose $\varepsilon = \min (\varepsilon_1, \varepsilon_2)$ and consider a second profile $M^{\prime k} \gg M^k$ with $M^{\prime k}_i  = M^k_i  + \varepsilon$ for $i=1, \ldots , n$.  All the endpoints in $M^{\prime k}$ have increased, but none has moved past $q_\ell^k$. Hence, $\texttt{emed} (M^{\prime k}; \textbf{q})  = q_\ell^k$, contradicting strict responsiveness.
\end{proof}

\section{Statement about Conflict of Interest}

The authors have no relevant financial or non-financial interests to disclose. They certify that they have no affiliations with or involvement in any organization or entity with any financial interest or non-financial interest in the subject matter or materials discussed in this manuscript.

\appendix
\section{Appendix}

\subsection{Aggregation on ordinal scales}\label{app1} 

We summarize some results in the theory of aggregation on ordinal scales; see  \citet[Chapter~8]{gra09}. Assume that $X$ is  an open interval from the linearly ordered set $(\RR,<)$. Let $f : X^n \rightarrow X$ be a univariate $n$-ary function that aggregates $n$ independent variables on the same ordinal scale into a single value on the same ordinal scale. 

Let $N = \{ 1, 2, \ldots , n \}$ and denote respectively by $\vee$ and $\wedge$ the operator $\max$ and $\min$. A function $f: X^n \rightarrow X$ is a \emph{lattice polynomial} if there exists a nonconstant set function $\alpha : 2^N \rightarrow \{0, 1\}$ with $\alpha (\emptyset) = 0$ such that
$$f(\mathbf{x}) = \bigvee_{\substack{{A \subseteq N} \\ \alpha (A) = 1}} \bigwedge_{\substack {i \in A}} x_i$$
Lattice polynomials admit different representations; see \cite{cha07}. We are mainly interested in the following result \citep[Cor.~4.2]{mar02}. 
\begin{prop}
Suppose that $X$ is open. Then $f : X^n \rightarrow X$ is a stable and continuous function if and only if it is a lattice polynomial.
\end{prop} 
If one adds unanimity (also known as idempotency) to continuity, then the characterization of lattice polynomials holds even if $X$ is not open.
\begin{prop}\label{prop10}
The function $f : X^n \rightarrow X$ is unanimous, stable and continuous if and only if it is a lattice polynomial.
\end{prop} 
Adding anonymity (also known as symmetry) leads to a characterization of order statistics; see \cite{ovc96} and \cite{mar02}.
\begin{prop}\label{prop11}
The function $f : X^n \rightarrow X$ is unanimous, anonymous, stable and continuous if and only if it is an order statistic.
\end{prop} 

Theorem~1 deals with a separable aggregation rule $f(\textbf{s})=\left( f^1(\textbf{s}), \ldots, f^m(\textbf{s}) \right)$ where each component function satisfies~Proposition~\ref{prop11}. A straightforward application of Proposition~\ref{prop11} implies that each function $f^k$ is an order statistic. The argument given in the main text before the statement of Theorem~\ref{thm1} shows that strict consistency necessitates that the $p_i$'s are nondecreasing.

One may replace unanimity in Proposition~\ref{prop10} or~\ref{prop11} with responsiveness, also known as weak monotonicity: if $\textbf{s} \leqq \textbf{s} ^\prime$ then $f(\textbf{s} ) \le f (\textbf{s} ^\prime)$. See \citet[Proposition 3.3.(iv) and Thm.~4.3]{mar02}.
\begin{prop}\label{prop12}
The function $f : X^n \rightarrow X$ is responsive, stable and continuous if and only if it is a lattice polynomial.
\end{prop} 
\cite{cha07} extends Proposition~\ref{prop12} from the finite-dimensional to the infinite-dimensional case. 

\subsection{Independence of axioms}\label{app2} 

We show that the four axioms in Theorem~\ref{thm1} are independent by means of four counterexamples.

\smallskip\noindent 
1) Anonymity. Any dictatorship is unanimous, stable and continuous, but it is not anonymous.

\smallskip\noindent 
2) Stability. The mean rule is unanimous, anonymous and continuous, but it is not stable.

\smallskip\noindent 
3) Unanimity. A rule that sets $f^1(M^1)$ equal to $\inf X$ is anonymous, stable and continuous, but it is not unanimous.

\smallskip\noindent 
4) Continuity. Assume $n \ge 3$. A rule that sets $f^1(M^1)$ equal to $M^1_{(1)}$ if $s^1_i \ne s^1_j$ for all $i \ne j$ and otherwise sets it equal to $M^1_{(n)}$ is unanimous, anonymous, and stable, but it is not continuous.

\subsection{Strategy-proofness}\label{app3}

We begin with a property that constrains how an aggregation rule is affected when agent $i$ changes one endpoint.

\smallskip\noindent 
\textbf{$k$-Uncompromisingness.} Given $i$, consider $M$ and $M^\prime$ such that $M_j = {M^\prime}_j$ for all $j \neq i$. An aggregation rule is \emph{uncompromising} over the $k$-th endpoint if: either (a) $f^k(M) = f^k({M^\prime})$; or (b) both $f^k(M) \in \langle M^k_i, f^k(M^\prime) \rangle$ and (c) $f^k({M^\prime}) \in \langle {M^\prime}^k_i, f^k(M) \rangle $ hold.
\smallskip

When the aggregation rule $f$ is strategy-proof, for any preference $\succeq_i$ in $\cal P$ with peak $M_i$, $f(M) \in \langle M_i, f(M^\prime) \rangle$. (Otherwise, assuming a preference in $\cal P$ with $f(M^\prime) \succ_i f(M)$, Agent~$i$ could profitably misreport $M^\prime_i$.) Symmetrically, $f(M^\prime) \in \langle M^\prime_i, f(M) \rangle$. 

\begin{lemma}\label{lem1}
A strategy-proof aggregation rule is $k$-uncompromising for every agent $i$ and every endpoint $k$.
\end{lemma}
\begin{proof} Fix an agent $i$ and an endpoint $k$. Assume $M_j = M^\prime_j$ for all $j \neq i$. There are two cases: either $f^k(M) \in \langle M^k_i, {M^\prime}^k_i \rangle$ or $f^k(M) \not\in \langle M^k_i, {M^\prime}^k_i \rangle$. In the first case, because $f(M^\prime) \in \langle M^\prime_i, f(M) \rangle$, then $f^k({M^\prime}) \in \langle {M^\prime}^k_i, f^k(M) \rangle$, so (b) and (c) hold. 

In the second case, either $f^k(M) < \inf( M^k_i, {M^\prime}^k)$ or $f^k(M) > \sup( M^k_i, {M^\prime}^k)$. Suppose $f^k(M) < \inf( M^k_i, {M^\prime}^k)$. Because $f(M) \in \langle M_i, f(M^\prime) \rangle$, then $f^k(M^\prime) \le f^k(M) \le M_i^k$. And, because $f(M^\prime) \in \langle M^\prime_i, f(M) \rangle$, then $f^k(M) \le f^k(M^\prime) \le M^{\prime k}_i$. Therefore, (a) holds. 

Finally, suppose $f^k(M) > \sup( M^k_i, {M^\prime}^k)$. Because $f(M) \in \langle M_i, f(M^\prime) \rangle$, then $f^k(M) \le f^k(M^\prime)$.  And, because $f(M^\prime) \in \langle M^\prime_i, f(M) \rangle$, then $f^k(M^\prime) \le f^k(M)$. Again, (a) holds. 
\end{proof}

\begin{customthm}{\ref{prop8}}
Any strategy-proof aggregation rule is separable.
\end{customthm}
\begin{proof}
Fix an arbitrary $k$-th endpoint and assume that the aggregation rule $f$ is strategy-proof. We need to show that if $M^k_i={M^\prime}^k_i$ for all $i$, then $f^k(M) = f^k({M^\prime})$. For $\ell=0,1, \ldots, m$, define $M_\ell$ as the matrix where the first $m -\ell$ columns are copied by $M$ and the last $\ell$ columns from $M^\prime$:
$$M_\ell = \left[ M^1 \cdots M^{m-\ell} \, \big\vert \, (M^{\prime})^{m - \ell +1} \cdots (M^{\prime})^m \right]$$
Clearly, $M_0 = M$ and $M_m= M^\prime$. 

Using induction, we show that $f^k (M_{\ell -1}) = f^k (M_\ell)$ for $\ell = 1, \ldots , m$. By Lemma~\ref{lem1}, $f$ is $k$-uncompromising. Therefore, either (a) $f^k(M_{\ell -1}) = f^k({M_\ell})$; or (b) both $f^k(M_{\ell - 1}) \in \langle M^k_{\ell -1}, f^k(M_\ell) \rangle$ and (c) $f^k({M_\ell}) \in \langle {M_\ell}^k, f^k(M_{\ell -1}) \rangle$. If (a) holds, we are done. Otherwise, unpacking (b) and (c), either (b1) $M^k_{\ell -1} \le f^k(M_{\ell - 1}) \le f^k(M_\ell)$ or (b2) $M^k_{\ell -1} \ge f^k(M_{\ell - 1}) \ge f^k(M_\ell)$, and either (c1) ${M_\ell}^k \le f^k({M_\ell}) \le f^k(M_{\ell -1})$ or (c2) ${M_\ell}^k \ge f^k({M_\ell}) \ge f^k(M_{\ell -1})$. Pairing the second inequalities from (b1) and (c1), or from (b2) and (c2), $f^k(M_{\ell -1}) = f^k({M_\ell})$ follows. Analogously, chaining the second inequalities from (b1) and (c2), or from (b2) and (c1), $f^k(M_{\ell -1}) = f^k({M_\ell})$ follows because $M^k_{\ell - 1} = M^k_\ell$.
\end{proof}

\end{document}